\documentclass[preprint,3p,12pt]{elsarticle}
\usepackage[utf8]{inputenc}
\usepackage[T1]{fontenc}
\usepackage{tikz}
\pdfoutput=1
\usepackage[colorlinks=true,breaklinks=true,bookmarks=true,urlcolor=blue,
citecolor=blue,linkcolor=blue,bookmarksopen=false,draft=false]{hyperref}%
\usepackage{doi}
\usetikzlibrary{positioning,backgrounds,patterns,calc}
\usepackage{amsmath,amssymb,amsthm,mathtools}
\usepackage{enumerate}
\usepackage{booktabs}
\usepackage{tabularx}
\usepackage{multirow}
\usepackage{hyphenat}
\usepackage{tikz}
\usepackage{pgf}
\usepackage{microtype}
\usepackage[sort&compress,noabbrev, capitalize, nameinlink]{cleveref}

\usepackage[textsize=small]{todonotes}

\newcommand{\N}{\mathbb{N}}
\newcommand{\Z}{\mathbb{Z}}
\newcommand{\K}{\mathbb{K}}

\newcommand{\miPoSyCo}{\textsc{MinPSC}}

\newcommand{\mpsc}{\textsc{\miPoSyCo}}

\DeclareMathOperator{\dist}{dist}

\newcommand{\ceq}{\coloneqq}

\newtheorem{theorem}{Theorem}[section]
\newtheorem{proposition}[theorem]{Proposition}

\newtheorem{lemma}[theorem]{Lemma}
\newtheorem{corollary}[theorem]{Corollary}
\newtheorem{observation}[theorem]{Observation}

\theoremstyle{definition}

\newtheorem{definition}[theorem]{Definition}

\newtheorem{example}[theorem]{Example}

  \newcommand{\dectaskprob}[3]{
  \begin{center}\vspace{1em}
  \begin{minipage}{0.95\textwidth}
    \noindent\textsc{#1}
    \begin{description}
     \item[Input:] #2
     \item[Task:] #3
    \end{description}
    \end{minipage}
    \end{center}\vspace{1em}
  }

  \newcommand{\optprob}[3]{
    \pagebreak[3]
    \begin{center}\vspace{1em}
    \begin{minipage}{0.95\textwidth}
    \noindent\textsc{#1}
    \begin{description}
     \item[Input:] #2
     \item[Task:] #3
    \end{description}
    \end{minipage}
    \end{center}\vspace{1em}
  }

\newcommand{\norm}[2]{\left\lVert#1\right\rVert_{#2}}

\DeclareMathOperator{\sign}{sign}

\newcommand{\Q}{\mathbb{Q}}
\newcommand{\Qnn}{\ensuremath{\Q_{\geq 0}}}
\newcommand{\R}{\mathbb{R}}
  \DeclareMathOperator*{\argmax}{arg\,max}

\newcommand{\etal}{et~al.}

\newcommand{\prob}[1]{\textsc{#1}}
\newcommand{\wmod}[1]{\ensuremath{\widehat{#1}}}

\newcommand{\miposycoTsc}{\prob{Min-Power Symmetric Connectivity}}
\newcommand{\miposycoAcr}{\prob{MPSC}}

\newcommand{\classSty}[1]{\ensuremath{\text{#1}}}

\newcommand{\NP}{\classSty{NP}}

\newcommand{\lin}{linearizable}
\newcommand{\Klin}[1]{#1\hyp $\K$\hyp \lin}
\newcommand{\Qlin}[1]{#1\hyp $\Q$\hyp \lin}
\newcommand{\Zlin}[1]{#1\hyp $\Z$\hyp \lin}

\newcommand{\calF}{\mathcal{F}}

\usepackage{xargs} 
\newcommandx{\set}[2][1=1]{\ensuremath{\{#1,\ldots,#2\}}}
\newcommandx{\tlog}[3][1=,3=]{\log_{#1}^{#3}(#2)}

\newcommandx{\eqr}[4][3=\alpha,4=\Q^d]{\ensuremath{#1\sim_{#3}^{#4}#2}}
\newcommandx{\eqc}[3][2=\alpha,3=\Q^d]{\ensuremath{[#1]_{#2}^{#3}}}

\newcommand{\conc}[2]{\ensuremath{#1\circ#2}}

\newcommand{\uflpTsc}{Uncapacitated Facility Location}
\newcommand{\uflpAcr}{UFL}
\newcommand{\muflpTsc}{Metric Uncapacitated Facility Location}
\newcommand{\muflpAcr}{MUFL}

\newcommand{\mytitle}{%
Polynomial-Time Data Reduction for Weighted Problems Beyond Additive Goal Functions
}

\newcommand{\tuaffil}{%
Technische Universit\"at Berlin, Faculty~IV, Algorithmics and Computational Complexity, Berlin, Germany
}

\numberwithin{equation}{section}

\crefname{example}{Example}{Examples}
\Crefname{example}{Example}{Examples}
\crefname{problem}{Problem}{Problems}
\crefname{transformation}{Transformation}{Transformations}
\crefname{rrule}{Reduction Rule}{Reduction Rules}
\crefname{part}{Part}{Parts}
\crefname{proposition}{Proposition}{Proposition}
\Crefname{proposition}{Prop.}{Props.}
\crefname{theorem}{Theorem}{Theorems}
\Crefname{theorem}{Thm.}{Thms.}
\crefname{observation}{Observation}{Observations}

\journal{Discrete Applied Mathematics}

\begin{document}

\begin{frontmatter}

\title{\mytitle}

\author[ad1]{Matthias~Bentert}\ead{matthias.bentert@tu-berlin.de}
\author[ad2]{Ren\'e~van~Bevern\fnref{fn1}}\ead{rene.van.bevern@huawei.com}
\fntext[fn1]{The results obtained in this article are unrelated to the author's work at Huawei.
}
\author[ad1]{Till~Fluschnik\corref{cor1}}\ead{till.fluschnik@tu-berlin.de}
\cortext[cor1]{Corresponding author}
\author[ad1]{Andr\'e~Nichterlein}\ead{andre.nichterlein@tu-berlin.de}
\author[ad1]{Rolf~Niedermeier}

\address[ad1]{\tuaffil{}}
\address[ad2]{Huawei Technologies Co., Ltd.}

\begin{abstract}
  Dealing with NP-hard problems,
  kernelization is a fundamental notion
  for polynomial\hyp time data reduction with performance guarantees:
  in polynomial time,
  a problem instance
  is reduced to an equivalent instance
  with size upper-bounded by a function
  of a parameter chosen in advance.
  Kernelization for weighted problems
  particularly requires to also shrink weights.
  Marx and V\'egh [ACM Trans.\ Algorithms 2015]
  and Etscheid et al. [J.\ Comput.\ Syst.\ Sci.\ 2017]
  used 
  a technique of Frank and Tardos [Combinatorica 1987]
  to obtain polynomial\hyp size kernels for weighted problems,
  mostly
  with additive goal functions.
  We characterize the function types
  that the technique is applicable to,
  which turns out to contain many non-additive functions.
  Using this insight,
  we systematically obtain kernelization results for %
  natural problems
  in
  graph partitioning,
  network design,
  facility location,
  scheduling,
  vehicle routing,
  and computational social choice,
  thereby improving and generalizing
  results from the literature. %
\end{abstract}

\begin{keyword}
NP\hyp hard problems\sep 
problem kernelization\sep 
weight reduction\sep
routing\sep 
scheduling\sep 
computational social choice\sep
partitioning

\end{keyword}

\end{frontmatter}
\thispagestyle{empty}
\section{Introduction}
\label{sec:ft!intro}

In the early eighties,
Gr\"otschel~\etal~\citep{GLS81} employed the famous ellipsoid method by Khachiyan
\citep{Kha79,Kha80} 
for solving the \prob{Weighted Independent Set} (WIS) problem:
Given an undirected graph~$G=(V,E)$ with vertex weights~$w\colon V\to \Q_+$,
find a set~$U\subseteq V$ such that~$U$ is an independent set and maximizes~$\sum_{v\in U} w(v)$.
Gr\"otschel~\etal~\citep{GLS81} proved WIS to be solvable in polynomial time on perfect graphs.
The running time of their algorithm,
however,
was only \emph{weakly} polynomial,
which led to the question whether
WIS on perfect graphs is
solvable in \emph{strongly} polynomial time.\footnote{Not to be confused with pseudo\hyp polynomial and
  polynomial running time, see,
  e.g.,
  Schrijver~\citep[Section~4.12]{Sch03} or 
  Gr{\"{o}}tschel et al.~\citep[Section~1.3]{GLS1988}.
  For strongly polynomial time,
  one requires from the algorithm
  to have space polynomial in the input size
  and that the number of elementary arithmetic and other operations
  executed by the algorithm does not depend
  on the sizes of the numbers in the input.}
In their seminal work,
Frank and Tardos~\citep{ft87} affirmatively answered this question by developing a (what we call) \emph{losing\hyp weight technique}.
Their technique employs a preprocessing algorithm that,
exemplified for WIS,
does the following:

\begin{example}[\prob{Weighted Independent Set}]
 \label[example]{ex:wis}
 In strongly polynomial time,
 compute vertex weights~$\wmod{w}$ such that %
  \begin{enumerate}[(a)]
  \item the encoding length of the maximum value of~$\wmod{w}$ is upper-bounded by a polynomial in the number of graph vertices,
  \end{enumerate}
  while preserving the relative quality of all solutions and non-solutions,
  that is,
  \begin{enumerate}[(b)]
  \item for every two (independent) sets~$U,U'\subseteq V$,
    it holds that $\sum_{v\in U} w(v)\geq \sum_{v\in U'} w(v)$ 
    if and only if $\sum_{v\in U} \wmod{w}(v)\geq \sum_{v\in U'} \wmod{w}(v)$.
  \end{enumerate}
\end{example}

Thus, WIS can be solved in strongly polynomial time on perfect graphs
by first applying the losing\hyp weight technique 
and then the algorithm of Gr\"otschel \etal{}~\citep{GLS81}.
To the best of our knowledge, 
the
technique  
was used the first time in the context of parameterized algorithmics %
by Fellows~\etal~\citep{FLM+08},
where it was used
to obtain fixed-parameter algorithms running in polynomial space.
Marx and V\'egh~\citep{MV15} first observed the connection of the losing\hyp weight technique to polynomial\hyp time data reduction, %
namely kernelization:
intuitively,
in polynomial time,
a problem instance
is reduced to an equivalent instance
with size upper-bounded by a function
of a problem-specific parameter.
Notably,
their kernelization first increases the size of the instance and then introduces additional edge weights.
Marx and V\'egh~\citep{MV15} stated that %
 ``[...] this technique seems to be an essential tool for kernelization of
problems involving costs.''
Subsequently,
Etscheid~\etal~\citep{ekmr17}
and \citet{KK20}
used the technique to prove polynomial kernels for several weighted problems,
supporting Marx and V\'egh's statement.

In almost all problems studied by the four papers mentioned above,
the goal functions are additive set functions
(that is, functions~$f$ satisfying $f(A\uplus B)=f(A)+f(B)$ for sets~$A$ and~$B$).
In the two cases where they are not,
ad\hyp hoc adaptions of \citeauthor{ft87}'
theorem~\cite{ft87} are used.
We present a method of systematically recognizing
non\hyp additive functions 
(which are not necessarily set functions)
to which the losing\hyp weight technique applies.

\paragraph{Our Contributions and Structure of this Work}

In \cref{sec:prelim-losing-weight},
we introduce basic notation
and give a brief introduction to
the losing\hyp weight technique. %
In~\cref{sec:twoprobs-nonlineargf},
we show how to apply the losing\hyp weight technique to
two problems with non\hyp additive goal functions
in graph partitioning and network design.
In~\cref{sec:alphlin},
we characterize what we call $\alpha$-\emph{\lin{}} functions
to which \citeauthor{ft87}'~\citep{ft87} losing\hyp weight technique applies.
Intuitively,
the parameter~$\alpha$ associated with a \lin{} function specifies how ``far'' the function is from an everywhere-linear function.
We additionally provide some tools
that allow for a convenient computation of a linearizable function's $\alpha$-value.
In \cref{ssec:apps},
we exemplify the versatility of these tools using problems
from %
network design, %
facility location, %
scheduling, %
vehicle routing, %
and computational social choice. %

We complement or improve several results in the literature:
In \cref{ssec:qlin},
  we settle an open problem on the
  kernelizability of the \textsc{Min-Power Symmetric Connectivity} problem~\citep{bbnn17}.
  In \cref{sec:uflp},
  we show a problem kernel for the \textsc{Uncapacitated Facility Location} problem
  whose size is polynomially upper-bounded in the number of the vertices of the input graph.
  Previously, 
  only problem kernels with size exponentially upper-bounded
  in the optimal solution cost
  (which is usually larger than the number of vertices)
  were known~\citep{FF11}.
  In \cref{sec:sched}, we shrink weights in
  several classical scheduling problems.
  Polynomial problem kernels for scheduling problems
  are rare \citep{BMNW15,KK20,MB18}
  and shrinking weights will necessarily be an ingredient
  in kernels for weighted scheduling problems.
  In \cref{ssec:rpp},
  we generalize a kernelization result
  for the \textsc{Rural Postman Problem} \citep{BNSW15}
  to the \textsc{Min-Max $k$-Rural Postman Problem}.
  In \cref{ssec:pvc},
  we prove a theorem on polynomial kernelization for the \textsc{Power Vertex Cover} problem that has been stated without proof in the literature~\citep{ABEL18}.

\section{Preliminaries and the Losing-Weight Technique}
\label{sec:prelim-losing-weight}
\label{ssec:losing-weight}

\subsection{Basic Notation and Definitions}
An $n$-dimensional vector~$x\in S^n$ for some set~$S$ is interpreted as a column vector,
and we denote by~$x^\top$ its transpose.
For two vectors~$x=(x_1,\dots,x_n)\in S^n$ and~$y=(y_1,\dots,y_m)\in T^m$,
we denote by~$\conc{x}{y}\ceq (x_1,\dots,x_n,y_1,\dots,y_m)\in (S\cup T)^{n+m}$ the concatenation of~$x$ and~$y$.
The \emph{$\ell_1$-norm}
of a vector~$x\in \R^n$
is $\norm{x}{1} := \sum_{i=1}^n |x_i|$.
The~$\ell_\infty$-norm (also known as max-norm) of~$x$ is~$\norm{x}{\infty} := \max_{i\in\set{n}}|x_i|$.
For a number $x\in\R$, 
the signum of~$x$ is defined by~$\sign(x)\ceq 1$ if $x>0$, $\sign(x)\ceq 0$ if $x=0$, and $\sign(x)\ceq -1$ if $x<0$.

Let~$\Sigma$ be a finite alphabet.
A set~$P\subseteq \Sigma^*\times\N$
is called a \emph{parameterized problem}.
In an \emph{instance}~$(x,k)\in\Sigma^*\times\N$,
we call $x$~the \emph{input} and $k$~the \emph{parameter}.
\begin{definition}
 \label{def:kernel}
 A \emph{problem kernelization}
 for a parameterized problem~$P\subseteq \Sigma^*\times\N$
 is an algorithm that, 
 given an instance~$(x,k)$,
 computes in polynomial time an instance~$(x',k')$
 such that 
 \begin{enumerate}[(i)]
 \item 
   $(x,k)\in P$ if and only if~$(x',k')\in P$, and
 \item 
   $|x'|+k'\leq f(k)$ for some  computable function~$f$ only depending on~$k$. 
 \end{enumerate}
 We call~$f$ the \emph{size} of the \emph{problem kernel}~$(x',k')$.
 If $f\in k^{O(1)}$, 
 then we call the problem kernel \emph{polynomial}.
 \end{definition}

\subsection{A Useful and Central Equivalence Relation}

In this section,
with the goal in mind to replace any given weight vector~$w$ 
by a ``representative'' weight vector~$\wmod{w}$ with upper-bounded~$\norm{\wmod{w}}{1}$,
we define an equivalence relation on vectors over~$\K\in\{\Z,\Q\}$.
Its equivalence classes will be formed
by partitioning the space using hyperplanes
with coefficients from
\begin{align}
 \Z_{r} &:= \{\pm p\in\Z \mid p\in\{0,\ldots,r\}\}\subseteq \Z \text{ or} \label{eq:zr}\\
 \Q_{r} &:= \left\{\pm\frac{p}{q}\,\,\middle|\,\,p\in\{0,\ldots,r\}, q\in\{1,\ldots,r\}\right\}\subseteq \Q .\label{eq:qr}
\end{align}
Specifically,
we will say that two vectors~$u$ and~$v$ are equivalent 
if and only if for all vectors~$\beta$ from some specific subset of~$\K_r$,
their dot products~$\beta^\top u$ and~$\beta^\top v$ have the same signum.
Geometrically speaking,
$u$ and~$v$ are equivalent
if and only if for all vectors~$\beta$ from some specific subset of~$\K_r$,
there is no hyperplane~$\{x\mid \beta^\top x=0\}$ separating~$u$ and~$v$.
Formally:

\begin{definition}
 \label{def:eqr}
 Let~$\K\in\{\Z,\Q\}$ and~$r,d\in\N$.
 Then,
 the binary relation~$\eqr{}{}[r][\K^d]$ on~$\Q^d$ is given by
 \[ \eqr{w}{w'}[r][\K^d] \iff \text{$\forall\beta\in \K_r^d$ with~$\norm{\beta}{1}\leq r$ it holds that~$\sign(\beta^\top w)=\sign(\beta^\top w')$} .\]
 For every~$w\in\Q^d$,
 let~$\eqc{w}[r][\K^d]\ceq \{w'\in \Q^d\mid \eqr{w'}{w}[r][\K^d]\}\subseteq \Q^d$ be the class of~$w$ under~$\eqr{}{}[r][\K^d]$.
\end{definition}

\begin{example}\label{ex:eqc}
Consider~$\Q^2$ and~$r=2$.
Any two vectors fall into the same class under~$\eqr{}{}[2][\Z^2]$ 
if and only if 
they cannot be separated by vectors from~$\Z^2$ with entries in~$\{0,\pm 1\}$
(see~\cref{fig:ftpartition} for an illustration).
\end{example}

We prove next that the relation from~\cref{def:eqr} is an equivalence relation.
\begin{figure}[t!]
  \centering
  \begin{tikzpicture}

    \usetikzlibrary{calc}

    \def\xr{1}
    \def\yr{1}

    \tikzstyle{invn}=[inner sep=0pt]

    \node (xl) at (-4*\xr,0)[invn]{};
    \node (xr) at (4*\xr,0)[invn,label=0:{$w_1$}]{};
    \node (xz) at (0,0)[invn]{};
    \node (xt) at (0,4*\yr)[invn,label=90:{$w_2$}]{};
    \node (xb) at (0,-4*\yr)[invn]{};

    \node (xrt) at (xr|-xt)[invn]{};
    \node (xrb) at (xr|-xb)[invn]{};
    \node (xlt) at (xl|-xt)[invn]{};
    \node (xlb) at (xl|-xb)[invn]{};

    \draw[very thick,->,>=latex] (xb) to (xt);
    \draw[very thick,->,>=latex] (xl) to (xr);

    \def\xeps{0.1}
    \def\opc{0.15}

    \node at (xz)[fill=black,opacity=\opc,circle,label={[yshift=-2.25em,color=gray]$C_0$}]{};

    \newcommand{\conetr}[1]{\draw[opacity=\opc,fill=#1,rounded corners] ($(xt.center)+(2*\xeps,0)$) -- ($(xz.center)+(2*\xeps,4*\xeps)$) -- ($(xrt.center)+(-2*\xeps,0)$);}

    \newcommand{\conebr}[1]{\draw[opacity=\opc,fill=#1,rounded corners] ($(xr.center)+(0,2*\xeps)$) -- ($(xz.center)+(4*\xeps,2*\xeps)$) -- ($(xrt.center)+(0,-2*\xeps)$);}

    \newcommand{\coned}[1]{\draw[thick,dashed] (xz) to (xrt);
      \draw[opacity=\opc,fill=#1,rounded corners]  ($(xrt.center)+(-1*\xeps,0)$) -- ($(xz.center)+(\xeps,2*\xeps)$) -- ($(xz.center)+(2*\xeps,\xeps)$) -- ($(xrt.center)+(0,-1*\xeps)$);}

    \newcommand{\coneu}[1]{\draw[opacity=\opc,fill=#1,rounded corners]  ($(xt.center)+(-\xeps,0)$) -- ($(xz.center)+(-\xeps,1*\xeps)$) -- ($(xz.center)+(1*\xeps,1*\xeps)$) -- ($(xt.center)+(1*\xeps,0)$);}

    \newcommand{\spinit}[5]{
      \foreach \x/\y in {0/#2,90/#3,180/#4,270/#5}{\begin{scope}[transform canvas={xshift = -\xr*0cm,yshift=\yr*0cm,rotate=\x}]#1{\y};\end{scope}}
    }

    \spinit{\conetr}{green}{red}{blue}{yellow};
    \spinit{\conebr}{blue!60!magenta}{yellow!60!magenta}{green!60!magenta}{red!60!magenta};
    \spinit{\coned}{cyan}{magenta}{orange}{brown};
    \spinit{\coneu}{orange!60!red}{brown!60!red}{cyan!60!red}{magenta!60!red};

    \foreach[count=\i,evaluate=\i as \angle using (\i-1)*360/16] \j in {1,...,16}\node[] (node\i) at (\angle:3.5)[color=gray]{$C_{\j}$};

    \draw[very thick,dotted,gray] (2*\xr,-2*\xr) -- (2*\xr,2*\xr) -- (-2*\xr,2*\xr) -- (-2*\xr,-2*\xr) -- cycle;
    \node (w) at (2.4*\xr,3.4*\yr)[circle,fill=black,scale=0.5,draw,label=90:{$w\in\Q^d$}]{};
    \node (wh) at (0.9*\xr,2*\yr)[circle,fill=black,scale=0.5,draw,label=-90:{$\wmod{w}\in\Z^d$}]{};
    \draw[->,>=latex] (w) to (wh);

  \end{tikzpicture}
  \caption{Illustration of the equivalence classes~$C_0,C_1,\ldots,C_{16}$ regarding~$\eqr{}{}[r][\Z^d]$ partitioning~$\Q^d$ with~$d=2$ and~$r=2$.
  \cref{thm:FrankTardos} is exemplified with some~$w$ and~$\wmod{w}$,
  each of which belonging to the equivalence class~$C_4$,
  where the dotted rectangle illustratively frames all vectors fulfilling~\cref{thm:FrankTardos}\eqref{ft:i}.
  }
  \label{fig:ftpartition}
\end{figure}
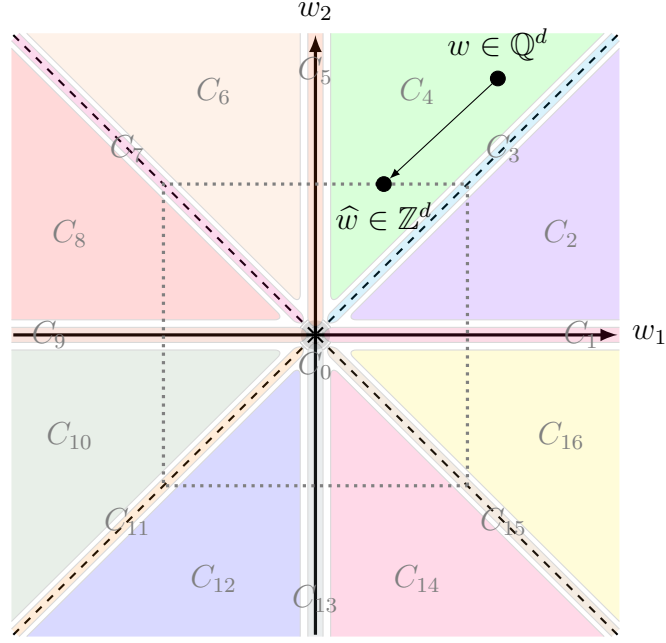

\begin{observation}\label{obs:eqc}
 For every~$\K\in\{\Z,\Q\}$ and~$r,d\in\N$,
 the relation~$\eqr{}{}[r][\K^d]$ on~$\Q^d$ is an equivalence relation.
\end{observation}

\begin{proof}
  Let~$\K\in\{\Z,\Q\}$ and~$r,d\in\N$. %
  Clearly, 
  $\eqr{w}{w}[r][\K^d]$ (reflexivity)
  and~$\eqr{w}{w'}[r][\K^d]\iff \eqr{w'}{w}[r][\K^d]$ (symmetry).
  Moreover,
  if~$\eqr{w}{w'}[r][\K^d]$ and~$\eqr{w'}{w''}[r][\K^d]$,
  then~$\eqr{w}{w''}[r][\K^d]$ (transitivity):
  For every $\beta\in \K_r^d$ with~$\norm{\beta}{1}\leq r$,
  one has~$\sign(\beta^\top w)=\sign(\beta^\top w')=\sign(\beta^\top w'')$.
\end{proof}

Next,
we prove some properties of~$\eqr{}{}[r][\K^d]$ and~$\eqc{\cdot}[r][\K^d]$.
The first property is the following:

\begin{observation}\label[observation]{obs:smallerClass}
 Let~$\K\in\{\Z,\Q\}$, $d\in \N$, and~$w\in\K^d$.
 For every~$r,r'\in \N$ with~$r\leq r'$ it holds that~$\eqc{w}[r][\K^d]\supseteq\eqc{w}[r'][\K^d]$.
\end{observation}

\begin{proof}
 Let~$\K\in\{\Z,\Q\}$, $d\in \N$,~$w\in\K^d$,
 and~$r,r'\in \N$ with~$r\leq r'$.
 Let~$w'\in \eqc{w}[r'][\K^d]$.
 We prove that~$w'\in\eqc{w}[r][\K^d]$.
 To this end,
 let~$\beta\in\K_{r}^d$ with~$\norm{\beta}{1}\leq r$.
 Note that~$\beta\in\K_{r'}^d\supseteq \K_{r}^d$ and~$\norm{\beta}{1}\leq r\leq r'$.
 Hence,
 $\sign(\beta^\top w)=\sign(\beta^\top w')$,
 and, therefore,
 $\eqr{w}{w'}[r][\K^d]$.
\end{proof}

Reconsider~\cref{ex:eqc} to exemplify~\cref{obs:smallerClass}. 
If we change~$r$ to one,
then,
for instance,
the union of the equivalence classes~$C_2$, $C_3$, and~$C_4$ forms an equivalence class,
say~$C$,
under~$\eqr{}{}[1][\Z^2]$.
Since~$w\in C_4$,
we have that~$w\in C$.

The second property is the following:

\begin{observation}
 \label[observation]{obs:signandorder}
 Let~$w=(w_1,\dots,w_d)\in\Q^d$, $d\in\N$, and~$\K\in\{\Z,\Q\}$. Then,
 \begin{enumerate}[(i)]
  \item for every~$r\geq 1$ and~$w'=(w_1',\dots,w_d')\in\eqc{w}[r][\K^d]$ it holds that~\(\sign(w_i)=\sign(w_i')\) for all~\(i\in\set{d}\);\label{obs:signandorder:sign}
  \item for every~$r\geq 2$ and~$w'=(w_1',\dots,w_d')\in\eqc{w}[r][\K^d]$ it holds that~\(\sign(w_i-w_j)=\sign(w_i'-w_j')\) for all~\(i,j\in\set{d}\).\label{obs:signandorder:order}
 \end{enumerate}
\end{observation}

Reconsider~\cref{ex:eqc} to illustrate~\cref{obs:signandorder}.
For instance,
we have that
for every vector~$w=(w_1,w_2)\in C_2$ it holds that~$w_1,w_2>0$ and~$w_2< w_1$,
for every vector~$w=(w_1,w_2)\in C_3$ it holds that~$w_1,w_2>0$ and~$w_1=w_2$,
and for every vector~$w=(w_1,w_2)\in C_4$ it holds that~$w_1,w_2>0$ and~$w_1< w_2$.

\subsection{Losing-Weight Technique}
Our work heavily relies on the following seminal result:
\begin{theorem}[Frank and Tardos {\citep[Section~3]{ft87}}]
  \label{thm:FrankTardos}
  On inputs~$w\in\Q^d$ and integer~$N$,
  one can compute in time polynomial in the encoding length of~$w$ and~$N$
  a vector~$\wmod{w}\in \Z^d$ with
  \begin{enumerate}[(i)]
   \item $\norm{\wmod{w}}{\infty}\leq 2^{4d^3}(N+1)^{d(d+2)}$ such that\label{ft:i}
   \item $\sign(w^\top b)=\sign(\wmod{w}^\top b)$
  for all~$b\in\Z^d$ with~$\norm{b}{1}\leq N$.
  \end{enumerate}
\end{theorem}

We briefly explain how
\cref{thm:FrankTardos}
relates to \cref{ex:wis}:
The vertex weights in \prob{Weighted Independent Set} can be interpreted as a vector~$w\in\Q^d$
with~$d:=|V|$.
Any two vertex subsets~$U,U'\subseteq V$
can be interpreted as vectors~$u,u'\in\{0,1\}^d$,
where $u_v=1$ if and only if~$v\in U$,
and $u'_v=1$ if and only if $v\in U'$.
Then,
$\sum_{v\in U} w_v = u^\top w$ and $\sum_{v\in U'} w_v=u'^\top w$.
With~$b:=u-u'$,
the statement of \cref{ex:wis}(b) can thus
be rewritten as
$b^\top w\geq 0 \iff b^\top \wmod{w}\geq 0$.
Since~$\norm{b}{1}\leq |V|$,
applying \cref{thm:FrankTardos} to~$w$ with $N:=|V|$
yields a new weight vector $\wmod{w}$
satisfying \cref{ex:wis}(a) 
and \cref{ex:wis}(b).

\cref{thm:FrankTardos} also works for decision rather than optimization problems.
Indeed,
the application to decision problems is a direct corollary,
first stated by Marx and V\'egh~\citep[Remark~3.15]{MV15} and then formalized by Etscheid~\etal~\citep{ekmr17},
thereby observing that the value given additionally along the description of the decision problem can be ``attached'' to the weight vector.

\begin{corollary}%
 \label{cor:FrankTardos}
  Given~$w\in\Q^d$, 
  $k\in\Q$,
  and~$N \in \N$,
  in time polynomial in the encoding length of~$w$, 
  $k$, 
  and~$N$,
  one can compute 
  a vector~$\wmod{w}\in \Z^d$~and an integer~$\wmod{k}\in\Z$ such that
  \begin{enumerate}[(i)]
   \item $\norm{\wmod{w}}{\infty},|\wmod{k}|\leq 2^{4(d+1)^3}(N+1)^{(d+1)(d+3)}$ and
    \item $w^\top b\leq k\iff \wmod{w}^\top b\leq \wmod{k}$
  for all~$b\in\Z^d$ with~$\norm{b}{1}\leq N-1$.
  \end{enumerate}
\end{corollary}

Whenever we are facing a weighted problem with an additive goal function,
that is, 
for example finding some set~$S$ such that~$\sum_{s\in S} w(s)$ is minimized (or maximized),
the application of \cref{thm:FrankTardos} is  immediate.
So it is for the well-known \prob{Knapsack} problem,
as proven by Etscheid~\etal~\citep{ekmr17},
giving the affirmative answer to the open question~\citep{CLF+14,FellowsGMS12} of whether \prob{Knapsack} admits a kernel of size polynomial in the number if items:
\begin{example}
  \label{ex:ft!knapsack}
 Recall the \prob{Knapsack} problem:
 Given a set~$X=\set{n}$ of items with
 weights $w\colon X\to\Q$ and values~$v\colon X\to\Q$,
 and rational numbers~$k,\ell\in \Q$,
 the question is
 whether there is a subset~$S\subseteq X$ of items such that~$\sum_{i\in S} w(i)\leq k$ and~$\sum_{i\in S} v(i)\geq \ell$.
 Let~$w$ and~$v$ be interpreted as $n$-dimensional vectors with~$w_i\ceq w(i)$ and~$v_i\ceq v(i)$.
 Applying \cref{cor:FrankTardos} 
 once with input~$w$,
 $k$, 
 and~$N\ceq n+1$,
 and once with input
 $v$,
 $\ell$, 
 and~$N$,
 (where~$d=n$ in each application) 
 yields an equivalent instance of~\prob{Knapsack} where the weights and values are of encoding-length polynomial in~$n$.
 Hence,
 this yields a problem kernel of size polynomial in~$n$.
\end{example}

\subsection{Losing-Weight Technique and our Equivalence Relation Combined}

Note that~\cref{thm:FrankTardos}(ii) is equivalent to
$\wmod{w}$ being contained in the equivalence class~$\eqc{w}[N][\Z^d]$ of~$w$.
Hence,
\cref{thm:FrankTardos},
given a $d$-dimensional vector~$w$ and any positive integer~$N$,
efficiently computes an integral representative~$\wmod{w}$ from $w$'s equivalence class where each entry can be upper-bounded by some number only depending on~$d$ and~$N$
(see~\cref{fig:ftpartition} for an illustrative example).
Consequently,
we can rephrase \cref{thm:FrankTardos} with respect to our equivalence relation 
as follows:
\begin{theorem}[\cref{thm:FrankTardos} rephrased]
  \label{thm:FrankTardosRephr}
  On inputs~$w\in\Q^d$ and integer~$N$,
  in time polynomial in the encoding length of~$w$ and~$N$
  one can compute 
  a vector~$\wmod{w}\in \Z^d\cap \eqc{w}[N][\Z^d]$ with $\norm{\wmod{w}}{\infty}\leq 2^{4d^3}(N+1)^{d(d+2)}$.
\end{theorem}
For convenience,
we will refer to~\cref{thm:FrankTardosRephr}
(instead of~\cref{thm:FrankTardos})
in the remainder of this work.

\section{Two Case Studies with Non-Additive Goal Functions}
\label{sec:twoprobs-nonlineargf}

In this section,
we show two applications
of \cref{thm:FrankTardosRephr}
to optimization problems with non\hyp additive goal functions.
In \cref{ex:wis} (\prob{Weighted Independent Set})
and \cref{ex:ft!knapsack} (\prob{Knapsack})
the used representation of the vectors has a one-to-one
correspondence to solution candidates: 
Any solution candidate to WIS or \prob{Knapsack} is a set of vertices or items, 
respectively. 
Such a set can clearly be represented with a (binary) vector 
and \emph{every} (binary) vector represents a solution candidate.
Yet,
is the second requirement needed?
In several of the applications that we are going to present, 
this is in fact not the case. 
Our core idea is hence as follows:
We still require that every solution candidate can be represented as a vector,
however, 
we do \emph{not} require every vector to represent a solution candidate. 
Note that this is fine since 
$\eqr{}{}[][]$
holds for all vectors~$b$ from the vector space containing vectors representing solution candidates, 
and thus, 
also for all vectors that do represent solution candidates.
We next exemplify our idea using two problems
with non\hyp additive goal functions
and formalize them later in \cref{sec:alphlin}.

\subsection[The Case of SSE]{The Case of \prob{Small Set Expansion}}
\label{ssec:SSE}

Consider the following graph partitioning problem,
which was studied %
in the context of
bicriteria approximation~\citep{BFK+14} and the unique games conjecture~\citep{RS10}.

\dectaskprob{\prob{Small Set Expansion} (SSE)}
  {An undirected graph~$G = (V,E)$ with edge weights~$w\colon E\to \Q_+$.}%
  {Find a non-empty subset~$S\subseteq V$ of size at most~$|S|\leq n/2$ that minimizes
    \begin{align}
      \frac{1}{|S|} \sum_{e\in (S,V\setminus S)} w(e),\label{gfsse}
    \end{align}
    where~$(S,V\setminus S)$ denotes the set of all edges with exactly one endpoint in~$S$.
}%

The goal function's value for a vertex set~$S$ can be represented by~$w^\top s$ for a \emph{fractional} vector~$s\in \{0,1/|S|\}^{|E|}$, 
where~$w$ is interpreted as vector and an entry of~$s$ is non-zero if and only if the corresponding edge is in the edge cut~$(S,V\setminus S)$.
Fractional numbers,
however,
are not captured by \cref{thm:FrankTardosRephr}.
Yet, 
with a scaling argument we can derive the following analog to~\cref{thm:FrankTardosRephr} dealing with fractional numbers:

\begin{proposition}%
  \label[proposition]{thm:FrankTardosRational}
  On input~$w\in\Q^d$ and integer~$r\in\N$,
  one can compute in time polynomial in the encoding length of~$w$ and~$r$
  a vector~$\wmod{w}\in \Z^d\cap \eqc{w}[r][\Q^d]$
  with
    $\norm{\wmod{w}}{\infty}\leq 2^{4d^3}(r^2+1)^{r\cdot d(d+2)}$.
\end{proposition}

  \begin{proof}
  Apply~\cref{thm:FrankTardosRephr} with~$N=r!\cdot r$ 
  to obtain a vector~$\wmod{w}\in \Z^d\cap \eqc{w}[r!r][\Z^d]$ 
  with
  \[\norm{\wmod{w}}{\infty}\leq 2^{4d^3}(N+1)^{d(d+2)} = 2^{4d^3}(r!\cdot r+1)^{d(d+2)} \leq 2^{4d^3}(r^2+1)^{r\cdot d(d+2)}.\]
  It remains to prove that~$\wmod{w}\in \Z^d\cap \eqc{w}[r][\Q^d]$.
  Let~$b^*\in\Q_r^d$ with~$\norm{b^*}{1}\leq r$,
  and let~$b'\ceq r!\cdot b^*\in \Z^d_{r!\cdot r}$.
  Note that~$\norm{b'}{1}\leq r!\cdot r = N$.
  Thus,
  due to~\cref{thm:FrankTardosRephr},
  we have that
  \begin{align*}
      \sign(w^\top b^*) =  \sign(\wmod{w}^\top b^*) 
      &\iff \sign(r!\cdot w^\top b^*) =  \sign(r!\cdot\wmod{w}^\top b^*) \\
      &\iff \sign(w^\top (r!\cdot b^*)) =  \sign(\wmod{w}^\top (r!\cdot b^*)) \\
      &\iff \sign(w^\top b') =  \sign(\wmod{w}^\top b'). \qedhere
  \end{align*}
\end{proof}

From \cref{thm:FrankTardosRational},
we get the following.

\begin{lemma}
 \label{lem:magicweightsSSE}
 For an input instance~$(G=(V,E),w)$ of \prob{Small Set Expansion} with $n:=|V|$ and $m:=|E|$,
 in time polynomial in~$|(G,w)|$ 
 one can compute 
 an instance~$(G,\wmod{w})$
 of \prob{Small Set Expansion} such that
 \begin{enumerate}[(i)]
  \item $\norm{\wmod{w}}{\infty}\leq 2^{4m^3}\cdot (n^4\cdot m^2+1)^{n^2m^2(m+2)}$ and
  \item a solution $S\subseteq V$
    for $(G,w)$ is optimal if and only if it is optimal for~$(G,\wmod{w})$.
 \end{enumerate}
\end{lemma}

\begin{proof}
  Denote the edges of~$G$ as~$E=\{e_1,\ldots,e_m\}$ 
  and the weight functions~$w$ and~$\wmod{w}$ as vectors in~$\N^m$ such that~$w_i=w(e_i)$ and~$\wmod{w}_i=\wmod{w}(e_i)$ for all~$i\in\{1,\dots,m\}$.
 Apply \cref{thm:FrankTardosRational} with~$d=m$ and~$r=n^2 m$.
 Let~$S\subseteq V$ 
 and let $s\in\{0,1/|S|\}^m$ be the vector %
 such that $s_i\neq 0$ if and only if~$e_i\in (S,V\setminus S)$.
 Let~$S'\subseteq V$ be another set,
 and let~$s'\in\{0,1/|S'|\}^m$ with $s_i'\neq 0$ if and only if~$e_i\in (S',V\setminus S')$.
 Let~$b:=s-s'$. 
 Note that for each~$i\in\{1,\dots,m\}$ it holds that
 \[|s_i-s_i'|= \left|\frac{|S'|s_i}{|S'|} - \frac{|S|s_i'}{|S|}\right| \in \left\{0,\left|\frac{|S'|-|S|}{|S|\cdot |S'|}\right|,\frac{1}{|S|},\frac{1}{|S'|}\right\},\]
 and hence~$b\in \Q_{n^2}^m$ and~$\norm{b}{1}\leq n^2 m$.
 We thus get
 \begin{align*}
   s^\top w - (s')^\top w \leq 0 &\iff (s-s')^\top w\leq 0
   \\  &\stackrel{\mathclap{\text{\Cref{thm:FrankTardosRational}}}}{\iff}\  (s-s')^\top \wmod{w}\leq 0 
   \iff s^\top \wmod{w} - (s')^\top \wmod{w}\leq 0.\qedhere
 \end{align*}
\end{proof}

\subsection[The Case of MiPoSyCo]{The Case of \miposycoTsc{}}
\label{ssec:MiPoSyCo}
The previous case of \prob{Small Set Expansion} showed how \cref{thm:FrankTardosRephr}
can be applied to weighted sums.
Next we show how to deal with a non-additive functions involving maxima.
To this end, 
consider the following \NP-hard optimization problem from survivable network design~\citep{acm06,cps04},
which has also been studied in the context of parameterized complexity
with practical results \citep{BBN+xx,bbnn17} (same for the asymmetric case~\citep{BentertHHKN20}).
\dectaskprob{\miposycoTsc{} (\miposycoAcr{})}%
{A connected undirected graph~$G = (V,E)$ and edge weights~$w\colon E \to \N$.}%
{Find a connected spanning subgraph~$T=(V,F)$ of~$G$ that minimizes
\begin{align}
  \sum_{v \in V} \max_{\{u,v\}\in F}w(\{u,v\}).\label{gfmpsc}
\end{align}
}

Applying~\cref{thm:FrankTardosRephr} to the goal function~\eqref{gfmpsc}
is not obvious:
Let~$E=\{e_1,\ldots,e_m\}$ and the weight function~$w$ be represented as a vector in~$\N^m$ such that~$w_i=w(e_i)$.
Let~$b\in\{0,1\}^m$  be the vector representing the edge set~$F$ of a solution~$T=(V,F)$,
that is, 
$b_i=1$ if and only if~$e_i\in F$.
Then, the value~$w^\top b$ is \emph{not} equal to~$\sum_{v\in V} \max_{\{u,v\}\in F} w(\{u,v\})$.
See~\cref{fig:miposyco}(a) for an example.
\begin{figure}
 \centering
  \begin{tikzpicture}
      
      \tikzstyle{xnode}=[fill,circle,scale=1/2,draw]
      \tikzstyle{xedge}=[-]
      \tikzstyle{xxedge}=[-,ultra thick, blue]
      \def\xr{0.95}
      \def\yr{1}

      \begin{scope}[yshift=-\yr*1.25cm] 
        \node at (-0.5*\xr,2.6*\yr+0.25*\yr)[]{(a)};
        \node (a) at (0,0)[xnode,label=180:{$u$}]{};
        \node (b) at (1.25*\xr,0.75*\yr)[xnode,label=-90:{$v$}]{};
        \node (c) at (1.25*\xr,2.5*\yr)[xnode,label=90:{$x$}]{};
        \node (d) at (2.5*\xr,0)[xnode,label=0:{$y$}]{};
        
        \draw[xxedge] (a) -- node[above]{3}(b);
        \draw[xedge] (a) -- node[above left]{8}(c);
        \draw[xxedge] (b) -- node[left,yshift=-\yr*0.6em]{1}(c);
        \draw[xedge] (c) -- node[above right]{10}(d);
        \draw[xxedge] (b) -- node[above]{2}(d);
        \draw[xedge] (a) -- node[below]{7}(d);       
      \end{scope}

      \begin{scope}[xshift=\xr*8.75cm] 
      \node at (-4.75*\xr,2.6*\yr-\yr*1)[]{(b)};
      \node at (0*\xr,0*\yr){\small
        \setlength{\tabcolsep}{4pt}
        \renewcommand{\arraystretch}{0.85}
        \begin{tabular}{@{}rcccccc@{}}\toprule
                    & $\{u,v\}$ & $\{u,x\}$ & $\{u,y\}$ & $\{v,x\}$ & $\{v,y\}$ & $\{x,y\}$ \\\cmidrule{2-7}
          $w$ & 3 & 8 & 7 & 1 & 2 & 10 \\
          \midrule\midrule
          $u$ & \textcolor{blue}{\textbf{1}}  & 1 & 1  & 0                            & 0                            & 0 \\
          $v$ & \textcolor{blue}{\textbf{1}}  & 0 & 0  & 1                            & 1                            & 0 \\
          $x$ & 0                             & 1 & 0  & \textcolor{blue}{\textbf{1}} & 0                            & 1 \\
          $y$ & 0                             & 0 & 1  & 0                            & \textcolor{blue}{\textbf{1}} & 1 \\
          \midrule\midrule
          $b:=$ & \textcolor{blue}{\textbf{2}} & 0 & 0 & \textcolor{blue}{\textbf{1}} & \textcolor{blue}{\textbf{1}} & 0  \\
          \bottomrule
        \end{tabular}
      };
      \end{scope}
  \end{tikzpicture}
 \caption{Illustrative example for \miposycoAcr{} and the application of~\cref{thm:FrankTardosRephr}. 
 (a) depicts an edge-weighted undirected example graph with a connected spanning subgraph (indicated by thick edges) of edge-weight six, 
 and 
 (b) shows the incidence matrix of the graph in (a), 
 the vector~$w$ of edge-weights,
 and the vector~$b$ representing the solution from~(a) with goal function value~$w^\top b=9$.}
 \label{fig:miposyco}
\end{figure}
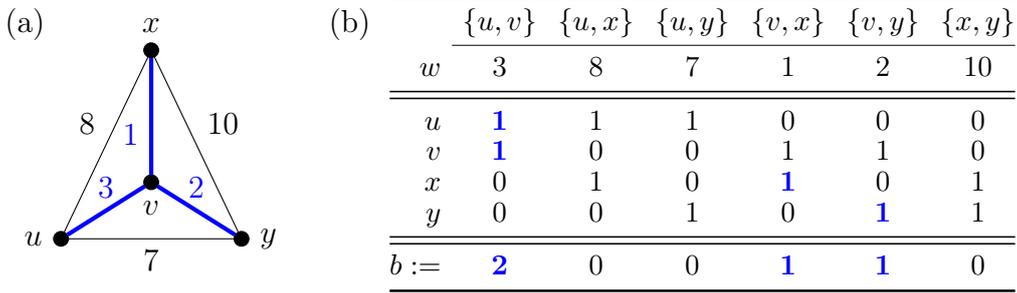

However,
we can circumvent this issue (arising from the $\max$-function in the goal function)
and still apply~\cref{thm:FrankTardosRephr}.
To this end,
observe that we only need to find a correct representation of a solution.
An edge~$e\in F$ contributes its weight to~\eqref{gfmpsc}
each time 
it appears in the (maximum in the) sum,
that is, either zero, one, or two times.
Hence,
a solution can be represented as vector~$b\in\{0,1,2\}^m$
such that the term $w(e_i)$~appears $b_i\in\{0,1,2\}$
times in the sum of the cost function regarding~$T=(V,F)$.
See~\cref{fig:miposyco}(b) for an example.
This change of the representation of a solution
only changes the domain of the vector~$b$, 
and hence increases the value of~$N$
in the application of~\cref{thm:FrankTardosRephr} by a factor of two.
Eventually,
we obtain:

\begin{lemma}
 \label{lem:magicweights}
 For an input instance~$(G=(V,E),w)$ of \miposycoAcr{} with $m:=|E|$,
 in time polynomial in~$|(G,w)|$ 
 one can compute
 an instance~$(G,\wmod{w})$
 of \miposycoAcr{} such that
 \begin{enumerate}[(i)]
  \item $\norm{\wmod{w}}{\infty}\leq 2^{4m^3}\cdot (2m+1)^{m(m+2)}$ and
  \item a connected subgraph $T=(V,F)$ of~\(G\)
    is an optimal solution for $(G,w)$ if and only if $T$ is an optimal solution for~$(G,\wmod{w})$.
 \end{enumerate}
\end{lemma}

\begin{proof}
  Denote the edges of~$G$ as~$E=\{e_1,\ldots,e_m\}$
  and %
  the weight functions~$w$ and $\wmod{w}$
  as (column) vectors in~$\N^m$ such that~$w_i=w(e_i)$ and~$\wmod{w}_i=\wmod{w}(e_i)$ for all~$i\in\{1,\dots,m\}$.
 We apply \cref{thm:FrankTardosRephr} with~$d=m$ and~$N=2m$ to the weight vector~$w$ and obtain the weight vector~$\wmod{w}$.
 \cref{thm:FrankTardosRephr} immediately implies statement~(i).
 Moreover, recall that~$\wmod{w}_i\geq 0$ for all~$i\in\{1,\ldots,m\}$ due to
 \cref{obs:signandorder}\eqref{obs:signandorder:sign},
 and hence~$(G,\wmod{w})$ is well-defined.

 Next, we prove statement~(ii).
 Let $T=(V,F)$ be a connected subgraph of~$G$
 and let $s\in\{0,1,2\}^m$ be an
 $m$-dimensional vector such that
 the term~$w(e_i)$ 
 appears $s_i$~times in~\eqref{gfmpsc}.
 Then,
 $\sum_{v \in V} \max_{\{u,v\}\in F}w(\{u,v\}) = s^\top w$.
 For a connected subgraph~$T'=(V,F')$ of~\(G\),
 let~$s'\in\{0,1,2\}^m$ be derived analogously so that
 the cost of~$T'$ is~$(s')^\top w$.
 Define~$b:=s-s'$.
 Note that $-2\leq b_i\leq 2$ for each~$i\in\set{m}$.
 Hence, $\norm{b}{1}\leq 2m=N$.
 Moreover, 
 due to~\cref{thm:FrankTardosRephr},
 we have that~$\wmod{w}\in\eqc{w}[2m][\Z^m]$,
 and hence
 \begin{align*} 
   s^\top w - (s')^\top w \leq 0 &\iff (s-s')^\top w\leq 0
   \\  &\stackrel{\mathclap{\norm{b}{1}\leq 2m}}{\iff}\ (s-s')^\top \wmod{w}\leq 0 
   \iff s^\top \wmod{w} - (s')^\top \wmod{w}\leq 0.
 \end{align*}
 Finally,
 note that due to
 \cref{obs:signandorder}\eqref{obs:signandorder:order},
 the goal function's values for both~$T$ and~$T'$
 with respect to~$\wmod{w}$
 are still correctly represented by~$s$ and~$s'$,
 that is,
 \begin{align*} 
      \sum_{v \in V} \max_{\{u,v\}\in F}\wmod{w}(\{u,v\}) &= s^\top \wmod{w} & \text{and} && %
      \sum_{v \in V} \max_{\{u,v\}\in F'}\wmod{w}(\{u,v\}) &= (s')^\top \wmod{w}.
     \qedhere                                                                                        %
 \end{align*}
\end{proof}

\section{Linearizable Functions}
\label{sec:alphlin}

In this section, 
we provide our central framework formalizing our key idea from the previous section.
Our framework bases on our notion of \emph{\lin{}} functions.
Before presenting the formal definition (see \cref{def:alphlin} below),
we recap the central insights from the previous section.

Our case studies for \prob{Small Set Expansion} and \miposycoTsc{} show that problems with non-additive goal functions still allow for an application of the losing\hyp weight technique.
A natural question is what characterizes these goal functions.
Both of our cases have in common that,
for any weight vector~$w$, 
the goal function's value for every solution~$s$ can be represented as~$b_s^\top w$ with~$b_s$ being a vector associated with~$s$.
Moreover, to apply the losing\hyp weight technique, we also need that if we change the weight vector to a ``smaller''  weight vector~$\wmod{w}$, then the goal function's value is still represented for solution~$s$ as~$b_s^\top \wmod{w}$ and vice versa 
(for this we used \cref{obs:signandorder}\eqref{obs:signandorder:order} in the proof of \cref{lem:magicweights}).
That is, we want that the value of solution~$s$ with respect to~$w$
is~$b_s^\top w$
if and only if the value of solution~$s$ with respect to~$\wmod{w}$
is~$b_s^\top \wmod{w}$.
Formally,
this is captured by the following.
(Let $\K_r$ with~$\K\in\{\Z,\Q\}$ be as defined in \eqref{eq:zr} and~\eqref{eq:qr}.)

\begin{definition}
 \label[definition]{def:alphlin}
  Let~$f\colon L\times \Q^d\to \Q$,
  where $L$ (here and in the following)
  is some arbitrary domain.
  We say that $f$ is \emph{\Klin{$\alpha$}}
  for some~$\alpha\in \N$
  if 
  for all~$w\in\Q^d$
  and for all~$x\in L$
  there exists a vector~$b_{x,w}\in\K_\alpha^d$ with~$\norm{b_{x,w}}{1}\leq \alpha$ 
  such that~$f(x,w')=b_{x,w}^\top w'$ for all $w'\in \eqc{w}[\alpha][\K^d]$.
\end{definition}

Intuitively,
an~\Klin{$\alpha$} function %
maps a solution 
(contained in the set~$L$) 
together with a weight vector to a number.
For a fixed weight vector
this number 
can be expressed for \emph{every} solution as the product of some vector representing the solution and the weight vector.
Moreover,
this representation of the solution is robust against exchanging weight vectors with any representative from its class.

We start with three basic properties of \lin{} functions.

\begin{observation}
	\label[observation]{obs:1linweightfcts}
	For any set~$X = \{x_1, x_2, \ldots, x_d\}$
  and any weight vector~$w=(w_1,w_2,\allowbreak\dots,w_d)\in\Q^d$,
	the function~$f\colon X \times \Q^{d} \to \Q$, $(x_i, w) \mapsto w_i$ %
  is \Klin{$1$} for every~$\K\in\{\Z,\Q\}$.
\end{observation}

\begin{observation}
	\label[observation]{obs:geqalphalin}
	Let~$\K\in\{\Z,\Q\}$.
	If~$f$ is~\Klin{$\alpha$}, then~$f$ is~\Klin{$\alpha'$} for all~$\alpha'\in\N$ with~$\alpha'\geq \alpha$.
\end{observation}

\begin{proof}
 Let~$w\in\Q^d$. %
 We know that, for all~$x\in L$, 
 there exists a vector~$b_x\in \K_\alpha^d\subseteq \K_{\alpha'}^d$ with~$\norm{b_x}{1}\leq \alpha\leq \alpha'$ 
 such that~$f(x,w')=b_x^\top w'$ for all~$w'\in \eqc{w}[\alpha][\K^d]\supseteq \eqc{w}[\alpha'][\K^d]$ 
 (recall~\cref{obs:smallerClass} for the latter inclusion).
\end{proof}

\begin{lemma}
 \label[lemma]{lem:ft!factor}
 Let~$\K\in\{\Z,\Q\}$,
 $f,f^*\colon L\times\Q^d\to \Q$, 
 and~$c\colon L\to\K_n\setminus\{0\}$, 
 where~$n\in \N$.
 If~$f$ is \Klin{$\alpha$},
 then $f^*(x,w)=c(x)\cdot f(x,w)$ is~\Klin{$n\alpha$}.
\end{lemma}

\begin{proof}
 Let~$w\in\Q^d$. %
 Since~$f$ is~\Klin{$\alpha$},
 we know that for every~$x\in L$ there exists a vector~$b_{x,w}\in\K_\alpha^d$ with~$\norm{b_{x,w}}{1}\leq \alpha$
 such that~$f(x,w')=b_{x,w}^\top w'$ for all~$w'\in\eqc{w}[\alpha][\K^d]$.
 Let~$b_{x,w}^*\ceq c(x)\cdot b_{x,w}$.
 We have that~$b_{x,w}^*\in \K_{n\alpha}^d$ and~$\norm{b_{x,w}^*}{1}\leq n\alpha$.
 Thus,
 for any $w'\in\eqc{w}[n\alpha][\K^d]\subseteq \eqc{w}[\alpha][\K^d]$,
 it holds that
 \[ f^*(x,w') = c(x)\cdot f(x,w') = c(x)\cdot b_{x,w}^\top w' = (c(x)\cdot b_{x,w})^\top w' = (b_{x,w}^*)^\top w'. \qedhere\]
\end{proof}

Next,
we prove next that the losing\hyp weight technique
applies to linearizable functions.
We first discuss $\Z$-\lin{} functions,
and afterwards $\Q$-\lin{} functions.

\subsection[Z-linearizable Functions]{$\Z$-\lin{} Functions}
\label{ssec:zlin}

The losing\hyp weight technique
applies to $\Z$-\lin{} functions as follows.

\begin{theorem}%
 \label{prop:ft!alphaZlin}
 Let~$f\colon L\times \Q^d\to \Q$
 be an \Zlin{$\alpha$} function,
 and let~$w\in \Q^d$, $k\in\Q$.
 Then 
 in time polynomial in the encoding length of~$w$, $k$, and~$\alpha$,
 one can compute 
 a vector~$\wmod{w}\in\Z^d$ and an integer~$\wmod{k}\in\Z$ such that
 \begin{enumerate}[(i)]
  \item $\norm{\wmod{w}}{\infty},|\wmod{k}|\leq 2^{4(d+1)^3}(2\alpha+1)^{(d+1)(d+3)}$,\label{prop:ft!alphaZlin:bound}
  \item $f(x,w)\geq f(y,w) \iff f(x,\wmod{w})\geq f(y,\wmod{w})$ for all~$x,y\in L$, and\label{prop:ft!alphaZlin:opt}
  \item $f(x,w)\geq k \iff f(x,\wmod{w})\geq \wmod{k}$ for all~$x\in L$.\label{prop:ft!alphaZlin:dec}
 \end{enumerate}
\end{theorem}

\begin{proof}
  Apply~\cref{thm:FrankTardosRephr} with~$N=2\alpha$ 
  to the vector $w\circ k$
  to obtain the concatenated vector~$\conc{\wmod{w}}{\wmod{k}}$ with
  \begin{align}
    \conc{\wmod{w}}{\wmod{k}}\in \eqc{\conc{w}{k}}[2\alpha][\Z^{d+1}] \label{eq:zlinwkmod}
  \end{align}
  and~$\|\conc{\wmod{w}}{\wmod{k}}\|_{\infty}\leq 2^{4(d+1)^3}(2\alpha+1)^{(d+1)(d+3)}$. 
  Thus,
  $\wmod{w}$ and~$\wmod{k}$ fulfill statement~\eqref{prop:ft!alphaZlin:bound}.
  Since~$f$ is~\Zlin{$\alpha$},
  by~\cref{def:alphlin},
  for every~$x,y\in L$ there are $b_x,b_y\in\Z_\alpha^d$ with~$\norm{b_x}{1},\norm{b_y}{1}\leq \alpha$ 
  such that~$f(x,w')=b_x^\top w'$ and $f(y,w')=b_y^\top w'$ for all~$w'\in\eqc{w}[\alpha][\Z^d]$.
  
  For statement~\eqref{prop:ft!alphaZlin:opt},
  let~$b:=b_x-b_y$.
  We have~$b \in\Z_{2\alpha}^d$ and~$\norm{b}{1}\leq 2\alpha$.
  Moreover
  \begin{align*}
    \sign(f(x,w)- f(y,w)) &= \sign((b_x-b_y)^\top w) \\ &\stackrel{\mathclap{\eqref{eq:zlinwkmod}}}{=}\quad \sign((b_x-b_y)^\top \wmod{w}) = \sign(f(x,\wmod{w}) - f(y,\wmod{w})),
  \end{align*}
  and hence
  \[ f(x,w)- f(y,w)\geq 0 \iff f(x,\wmod{w}) - f(y,\wmod{w}) \geq 0 .\] 
  
  For statement~\eqref{prop:ft!alphaZlin:dec},
  let~$b:=\conc{b_x}{(-1)}$.
  We have~$b \in\Z_{\alpha}^{d+1}\subseteq \Z_{2\alpha}^{d+1}$ and~$\norm{b}{1}\leq \alpha+1\leq 2\alpha$.
  Moreover,
  \begin{align*}
  \sign(f(x,w)-k) &= \sign( b_x^\top w- k) = \sign(b^\top(\conc{w}{k})) \\
                  &\stackrel{\mathclap{\eqref{eq:zlinwkmod}}}{=}\quad  \sign(b^\top(\conc{\wmod{w}}{\wmod{k}})) = \sign(b_x^\top \wmod{w}- \wmod{k}) = \sign(f(x,\wmod{w}) - \wmod{k})
  \end{align*}
  and hence
  \[ f(x,w)\geq k \iff f(x,\wmod{w})\geq \wmod{k} .\qedhere\] 
\end{proof}
  
Using \cref{prop:ft!alphaZlin},
we can shrink the weights in~\Zlin{$\alpha$} functions
so that their encoding length is polynomially upper-bounded in~$\alpha$ and the dimension~$d$.
For easy application of~\cref{prop:ft!alphaZlin},
we need to easily recognize \Zlin{$\alpha$} functions
and,
in particular,
to determine~$\alpha$.
To this end,
we show how
to
recognize an
\Zlin{$\alpha$} function
by simply looking at the functions it is composed of.
We subsequently demonstrate this using the example of~\miposycoTsc{} (\miposycoAcr{}).

\begin{lemma}%
 \label[lemma]{lem:ft!Z!summaxmin}
 Let~$f\colon L\times \Q^d\to \Q$ be a function.
 If~$f$ is \Zlin{$\alpha$}, 
 then the function~$f'\colon \{X\subseteq L\mid |X|\leq n\}\times \Q^d\to \Q$ with~$n\in\N$ and
\begin{enumerate}[(i)]
 \item $f'(X,w)=\sum_{x\in X} f(x,w)$ is~\Zlin{$n \alpha$};\label{alphZlin:sum}
 \item $f'(X,w)=\max_{x\in X} f(x,w)$ is~\Zlin{$2\alpha$};\label{alphZlin:max}
 \item $f'(X,w)=\min_{x\in X} f(x,w)$ is~\Zlin{$2\alpha$}.\label{alphZlin:min}
\end{enumerate}
\end{lemma}

  \begin{proof}
  \eqref{alphZlin:sum}:
    Let~$w\in \Q^d$ and~$X\subseteq L$ with~$|X|\leq n$. %
    Since~$f$ is \Zlin{$\alpha$},
    we know that,
    for all~$x\in X\subseteq L$ there is a~$b_{x,w}\in \Z_\alpha^d$ with~$\norm{b_{x,w}}{1}\leq \alpha$ 
    such that~$f(x,w')=b_{x,w}^\top w'$ for all~$w'\in\eqc{w}[\alpha][Z^d]$.
    Let~$b_X\ceq  (\sum_{x\in X} b_{x,w})\in \Z_{n\alpha}^d$.
    Let~$w'\in \eqc{w}[n\alpha][\Z^d]\subseteq \eqc{w}[\alpha][\Z^d]$. %
    We have that~$f'(X,w')=\sum_{x\in X} f(x,w')=\sum_{x\in X} b_{x,w}^\top w' = b_X^\top w'$. 
  
  \eqref{alphZlin:max}:
  Let~$w\in \Q^d$ and~$X\subseteq L$ with~$|X|\leq n$.
  We know that,
  for all~$x\in X\subseteq L$,
  there is a vector~$b_{x,w}\in \Z_\alpha^d$ such that~$f(x,w')=b_{x,w}^\top w'$ for all~$w'\in\eqc{w}[\alpha][\Z^d]$.
    Let~$z\in \argmax_{x\in X} b_{x,w}^\top w$ and~$b_X\ceq  b_{z,w}\in \Z_{\alpha}^d$.
    Let~$w'\in \eqc{w}[2\alpha][\Z^d]\subseteq \eqc{w}[\alpha][\Z^d]$. %
    For any~$y\in X$,
    let~$b\ceq b_{z}-b_y$.
    Note that~$b\in \Z_{2\alpha}^d$ and~$\norm{b}{1}\leq 2\alpha$,
    and hence~$\sign(b^\top w)=\sign(b^\top w')$.
    Thus, 
    we have that
    \begin{align*} 
      \sign(f(z,w)- f(y,w)) 
      &= \sign(b_{z}^\top w - b_{y}^\top w) =  \sign(b^\top w) \\
      &= \sign(b^\top w') = \sign(f(z,w')- f(y,w'))
    \end{align*}
    and hence it holds that~$f(z,w)\geq f(y,w) \iff f(z,w')\geq f(y,w')$.
    Finally, it follows that~$f'(X,w') = \max_{x\in X} f(x,w') = b_{z,w}^\top w' = b_X^\top w'$.
  
  \eqref{alphZlin:min}:
    Works analogously to~\eqref{alphZlin:max}.
  \end{proof}

\paragraph{Revisiting the Case of \miposycoTsc{}}
The goal function in \miposycoAcr{} is composed of a sum over maxima.
We proved that such a composition preserves linearizability.
We explain the use of our machinery for \miposycoAcr{}.
To this end,
rewrite the goal function as follows.
Let~$F_v:=\{e\in F\mid v\in e\}$ and~$\calF:=\{F_v\mid v\in V\}$.
Then, the goal function becomes
\[ h(\calF,w) = \sum_{F_v\in \calF} g(F_v,w)\quad\text{ with }\quad g(F,w)=\max_{e\in F} w(e).\]
Due to~\cref{obs:1linweightfcts},
$f(e,w)=w(e)$ is~\Zlin{1}.
Due to~\cref{lem:ft!Z!summaxmin}\eqref{alphZlin:max},
$g(F,w)=\max_{e\in F} f(e,w)$ is \Zlin{2}.
Finally,
due to \cref{lem:ft!Z!summaxmin}\eqref{alphZlin:sum} (with~$L=2^E$ and $n=|V|$),
$h(\calF,w) = \sum_{F_v\in \calF} g(F_v,w)$ is \Zlin{$2n$}.
Employing~\cref{prop:ft!alphaZlin},
we get in polynomial time a vector~$\wmod{w}\in \Q^m$ such that
   $\norm{\wmod{w}}{\infty}\in 2^{O(m^3\tlog{n})}$, and
   for any two connected subgraphs $T=(V,F)$ and $T'=(V,F')$ of~\(G\), 
   we have that
 \begin{align*}
    \sum_{v\in V} \max_{v\in e\in F} w(e) &\geq \sum_{v\in V} \max_{v\in e\in F'} w(e) \iff
    \sum_{v\in V} \max_{v\in e\in F} \wmod{w}(e) \geq \sum_{v\in V} \max_{v\in e\in F'} \wmod{w}(e).
 \end{align*}
We have thus reproven 
\cref{lem:magicweights}.
Moreover,
for the decision variant of~\mpsc{},
which asks whether there is a solution
of at most a given cost~$k$,
with \cref{prop:ft!alphaQlin} on input~$(G,w,k)$,
we immediately obtain a polynomial kernel.

\begin{proposition}
  \label[proposition]{thm:miposyco-small-weights}
  \miposycoTsc{} admits a polynomial kernel with respect to the number of vertices.
\end{proposition}
In previous work~\citep{BBN+xx,bbnn17}, 
we developed a partial kernel, that is, an algorithm that maps any instance of \miposycoAcr{}
to an equivalent instance where the number of vertices and edges, 
yet not necessarily the edge weights,
are polynomially upper-bounded in the feedback edge number.\footnote{The %
  smallest number of edges to be removed to transform a graph into a forest.}
Finding a polynomial kernel regarding the feedback edge number was an open problem.
Given the partial kernel~\citep{BBN+xx,bbnn17},
\cref{thm:miposyco-small-weights} yields the following affirmative answer.

\begin{corollary}
  \label{thm:feskernel}
  \miposycoTsc{} admits a polynomial kernel with respect to the feedback edge number of the input graph.
\end{corollary}

We will revisit the case of~\prob{Small Set Expansion (SSE)} in the next section,
using (an analog of) \cref{prop:ft!alphaZlin} for $\Q$-\lin{} functions.

\subsection[Q-linearizable Functions]{$\Q$-linearizable Functions}
\label{ssec:qlin}

In this section,
we give an analog of~\cref{prop:ft!alphaZlin} for
$\Q$-\lin{} functions and revisit the case of~\prob{Small Set Expansion}.
The analog of~\cref{prop:ft!alphaZlin} for
$\Q$-\lin{} functions is as follows.

\begin{theorem}
 \label{prop:ft!alphaQlin}
 Let~$f\colon L\times \Q^d\to \Q$
 be an \Qlin{$\alpha$} function,
 and let~$w\in \Q^d$, $k\in \Q$.
 Then,
 in time polynomial in the encoding length of~$w$, 
 $k$, 
 and~$\alpha$,
 one can compute 
 a vector~$\wmod{w}\in\Z^d$ and an integer~$\wmod{k}\in\Z$ such that
 \begin{enumerate}[(i)]
  \item $\norm{\wmod{w}}{\infty},|\wmod{k}|\leq 2^{4(d+1)^3}(4\alpha^4+1)^{2\alpha^2\cdot (d+1)(d+3)}$,\label{prop:ft!alphaQlin:bound}
  \item $f(x,w)\geq f(y,w) \iff f(x,\wmod{w})\geq f(y,\wmod{w})$ for all~$x,y\in L$, and\label{prop:ft!alphaQlin:opt}
  \item $f(x,w)\geq k \iff f(x,\wmod{w})\geq \wmod{k}$ for all~$x\in L$.\label{prop:ft!alphaQlin:dec}
  \end{enumerate}
\end{theorem}

\begin{proof}
  Apply~\cref{thm:FrankTardosRational} with~$r=2\alpha^2$
  to the concatenated vector~$\conc{w}{k}$
  to obtain the concatenated vector
  \begin{align}
    \conc{\wmod{w}}{\wmod{k}}\in\eqc{\conc{w}{k}}[2\alpha^2][\Q^{d+1}].\label{eq:Qeqcwkmod}
  \end{align}
 with~$\|\conc{\wmod{w}}{\wmod{k}}\|_{\infty}\leq 2^{4(d+1)^3}(4\alpha^4+1)^{2\alpha^2\cdot (d+1)(d+3)}$.
 Hence,
 $\wmod{w}$ and~$\wmod{k}$ fulfill statement~\eqref{prop:ft!alphaQlin:bound}.
 Since~$f$ is~\Qlin{$\alpha$},
 by~\cref{def:alphlin},
 for every~$x,y\in L$
 there are $b_{x,w},b_{y,w}\in\Q_\alpha^d$ with~$\norm{b_{x,w}}{1}\leq \alpha$ and~$\norm{b_{y,w}}{1}\leq \alpha$ 
 such that~$f(x,w')=b_{x,w}^\top w'$ and $f(y,w')=b_{y,w}^\top w'$ for all~$w'\in\eqc{w}[\alpha]\supseteq \eqc{w}[2\alpha^2]$.
 
 For statement~\eqref{prop:ft!alphaQlin:opt},
 let~$b:=b_{x,w}-b_{y,w}$.
 We have that~$b \in\Q_{2\alpha^2}^d$ and~$\norm{b}{1}\leq 2\alpha\leq 2\alpha^2$.
 Moreover,
 \begin{align*}
  \sign(f(x,w)- f(y,w)) &= \sign((b_x-b_y)^\top w) \\
                        &\stackrel{\mathclap{\eqref{eq:Qeqcwkmod}}}{=}\quad \sign((b_x-b_y)^\top \wmod{w}) = \sign(f(x,\wmod{w}) - f(y,\wmod{w})),
 \end{align*}
 and hence \[f(x,w) \geq f(y,w) \iff  f(x,\wmod{w}) \geq f(y,\wmod{w}).\]
 
 For statement~\eqref{prop:ft!alphaQlin:dec},
 let~$b:=\conc{b_x}{(-1)}$.
 We have that~$b \in\Q_{\alpha}^{d+1}\subseteq \Q_{2\alpha^2}^{d+1}$ and~$\norm{b}{1}\leq \alpha+1\leq 2\alpha^2$.
 Moreover,
 \begin{align*}
  \sign(f(x,w)-k) &= \sign( b_x^\top w- k) = \sign(b^\top(\conc{w}{k})) \\
                  &\stackrel{\mathclap{\eqref{eq:Qeqcwkmod}}}{=}\quad  \sign(b^\top(\conc{\wmod{w}}{\wmod{k}})) = \sign(b_x^\top \wmod{w}- \wmod{k}) = \sign(f(x,\wmod{w}) - \wmod{k}),
  \end{align*}
  and hence
  \[ f(x,w)\geq k \iff f(x,\wmod{w})\geq \wmod{k} .\qedhere\] 
\end{proof}

Next,
we present an analog of~\cref{lem:ft!Z!summaxmin} for~$\Q$-\lin{} functions.
It turns out that
composing~$\Q$-\lin{} functions introduces larger $\alpha$-values compared to~$\Z$-\lin{} functions.

\begin{lemma}%
 \label[lemma]{lem:ft!Q!summaxmin}
 Let~$f:L\times \Q^d\to \Q$ be a function.
 If~$f$ is \Qlin{$\alpha$}, 
 then the function~$f':\{X\subseteq L\mid |X|\leq n\}\times \Q^d\to \Q$ with~$n\in \N$ and
\begin{enumerate}[(i)]
	\item $f'(X,w)=\sum_{x\in X} f(x,w)$ is~\Qlin{$\alpha!n \alpha$};\label{alphQlin:sum}
	\item $f'(X,w)=\max_{x\in X} f(x,w)$ is~\Qlin{$2\alpha^2$};\label{alphQlin:max}
	\item $f'(X,w)=\min_{x\in X} f(x,w)$ is~\Qlin{$2\alpha^2$}.\label{alphQlin:min}
\end{enumerate}
\end{lemma}

\begin{proof}
  \eqref{alphQlin:sum}:
    Let~$w\in \Q^d$ and~$X\subseteq L$ with~$|X|\leq n$. %
    Since~$f$ is \Qlin{$\alpha$},
    we know that for all~$x\in X\subseteq L$ there is a vector~$b_{x,w}\in \Q_\alpha^d$ with~$\norm{b_{x,w}}{1}\leq \alpha$ 
    such that~$f(x,w')=b_{x,w}^\top w'$ for all~$w'\in\eqc{w}[\alpha][Q^d]$.
    Let~$b_X\ceq  (\sum_{x\in X} b_{x,w})\in \Q_{\alpha!n \alpha}^d$.
    Let~$w'\in \eqc{w}[\alpha!n \alpha][\Q^d]\subseteq \eqc{w}[\alpha][\Q^d]$. %
    We have that~$f'(X,w')=\sum_{x\in X} f(x,w')=\sum_{x\in X} b_{x,w}^\top w' = b_X^\top w'$. 
  
  \eqref{alphQlin:max}:
    Let~$w\in \Q^d$ and~$X\subseteq L$ with~$|X|\leq n$. %
    Since~$f$ is \Qlin{$\alpha$},
    we know that for all~$x\in X\subseteq L$ there is a vector~$b_{x,w}\in \Q_\alpha^d$ such that~$f(x,w')=b_{x,w}^\top w'$ for all~$w'\in\eqc{w}[\alpha][\Z^d]$.
    Let~$z\in \argmax_{x\in X} b_{x,w}^\top w$ and~$b_X\ceq  b_{z,w}\in \Q_{\alpha}^d$.
    Let~$w'\in \eqc{w}[2\alpha^2][\Q^d]\subseteq \eqc{w}[\alpha][\Q^d]$. %
    For any~$y\in X$,
    let~$b\ceq b_{z}-b_y$.
    Note that~$b\in \Q_{2\alpha^2}^d$ and~$\norm{b}{1}\leq 2\alpha\leq 2\alpha^2$,
    and hence~$\sign(b^\top w)=\sign(b^\top w')$.
    Thus, 
    \begin{align*} 
      \sign(f(z,w)- f(y,w)) &= \sign(b_{z}^\top w - b_{y}^\top w) =  \sign(b^\top w) \\
                            &= \sign(b^\top w') = \sign(f(z,w')- f(y,w'))
    \end{align*}
    and thus it holds that~$f(z,w)\geq f(y,w) \iff f(z,w')\geq f(y,w')$.
    It follows that~$f'(X,w') = \max_{x\in X} f(x,w') = b_{z,w}^\top w' = b_X^\top w'$.
  
  \eqref{alphQlin:min}:
    Works analogously to~\eqref{alphQlin:max}.
\end{proof}

The framework for~$\Z$-\lin{} functions  allows for ``chaining up sums''
while keeping $\alpha$ polynomially bounded.
Note that this is in general not the case for~$\Q$-\lin{} functions when applying~\cref{lem:ft!Q!summaxmin}.
Although more restrictive, 
however,
the framework for~$\Z$-\lin{} functions is sufficient for MPSC and all upcoming examples except for
the following.

\paragraph{Revisiting the Case of \prob{Small Set Expansion}}

The goal function in SSE is a multiplication of a number and a sum.
By~\cref{lem:ft!factor},
we know that multiplication preserves linearizability.
Moreover,
by \cref{lem:ft!Q!summaxmin}\eqref{alphQlin:sum},
we know that the sum preserves linearizability.
So,
we are set to use our machinery for SSE.

Let~$E_S:=(S,V\setminus S)$ for all~$S\subseteq V$.
Let~$L:=\{(S,E_S)\mid S\subseteq V, 1\leq |S|\leq n/2\}$.
Let~$c\colon L\to \Q_n\setminus \{0\}, (S,E_S)\mapsto \frac{1}{|S|}$.
Then,
the goal function of SSE becomes
$h((S,E_S),w) = \frac{1}{|S|} \cdot g((S,E_S),w)$ with $g((S,E_S),w)=\sum_{e\in E_S} w(e)$.
By \cref{obs:1linweightfcts},
$f(e,w)=w(e)$ is \Qlin{1}.
Moreover,
by \cref{lem:ft!Q!summaxmin}\eqref{alphQlin:sum},
$g$ is~\Qlin{$m$}.
Finally,
due to \cref{lem:ft!factor},
$h$ is \Qlin{$n\cdot m$}.
Finally, employing~\cref{prop:ft!alphaQlin},
we can reprove
\cref{lem:magicweightsSSE}
and additionally obtain the following kernel.

\begin{proposition}
 \label[proposition]{prop:ssePKvertices}
  \prob{Small Set Expansion} admits a polynomial kernel with respect to the number of vertices.
\end{proposition}

\paragraph{Summary of our Framework}

We introduced \Klin{$\alpha$} functions (\cref{def:alphlin}) for~$\K\in\{\Z,\Q\}$.
Due to~\cref{lem:ft!Q!summaxmin,lem:ft!Z!summaxmin,lem:ft!factor},
we can easily recognize special types of \Klin{$\alpha$} functions by simply looking at their composition.
Further, 
we proved that %
the losing\hyp weight technique applies
to \Klin{$\alpha$} functions
(\cref{prop:ft!alphaQlin,prop:ft!alphaZlin}).
Thus,
for applying our framework, 
we offer the recipe in~\cref{fig:recipe}.

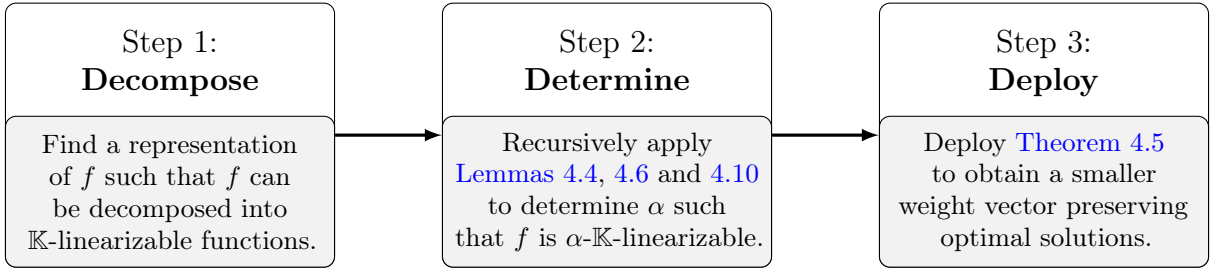
\begin{figure}\centering
  \begin{tikzpicture}

    \def\xr{1}
    \def\yr{1}
    \tikzstyle{mnode}=[rounded corners,minimum height=3.5cm,minimum width=4.75cm,draw]
    \tikzstyle{snode}=[align=center,anchor=south,rounded corners,minimum height=2.0cm,minimum width=4.75cm,text width=4cm,fill=gray!10!white,font=\footnotesize,draw]

    \newcommand{\stepnode}[4]{
    \node (s#2) at (#1)[mnode]{};
    \node at (s#2.north)[anchor=north,align=center]{\\ Step #2: \\ \textbf{#3}};
    \node at (s#2.south)[snode]{#4};
    }

    \stepnode{0,0}{1}{Decompose}{
      Find a representation of~$f$ such that~$f$ can be decomposed into $\K$-\lin{} functions.
    }

    \stepnode{5.75*\xr,0}{2}{Determine}{
      Recursively apply \cref{lem:ft!factor,lem:ft!Z!summaxmin,lem:ft!Q!summaxmin} 
      to determine~$\alpha$ such that~$f$ is~\Klin{$\alpha$}.
    }

    \stepnode{11.5*\xr,0}{3}{Deploy}{
      Deploy \cref{prop:ft!alphaZlin,prop:ft!alphaQlin}
      to obtain a smaller weight vector preserving optimal solutions.
    }

    \draw[very thick,->,>=latex] (s1) to (s2);
    \draw[very thick,->,>=latex] (s2) to (s3);
  \end{tikzpicture}
  \caption{Recipe DDD for applying our framework, 
    illustrated for a weighted problem seeking for a set that maximizes or minimizes some function~$f$.}
  \label{fig:recipe}
\end{figure}

We showed that any combination of sums, maxima, minima, 
and multiplication with (rational) numbers preserves linearizability.
In the next section, 
we show that we can also compose functions using
case distinctions on linearizable constraints (\cref{lem:lin-constraints}).
Finding compositions of further functions preserving linearizability remains a task for future work.

\section{Further Applications of the Losing-Weight Technique}
\label{ssec:apps}

In this section,
we provide further problems with \lin{} goal functions
and demonstrate how our framework applies to them via the recipe DDD.
The further problems stem from
network design,
facility location,
scheduling,
vehicle routing,
and computational social choice.

\subsection{Uncapacitated Facility Location.}
\label{sec:uflp}
The \prob{\uflpTsc} problem
is one of the most fundamental and well\hyp studied problems
in operations research \citep[Section 3.4]{LNG15}.
It has also been studied
in the context of parameterized complexity and data reduction
\citep{BTZ19,FF11}.

\optprob{\prob{\uflpTsc} (\prob{\uflpAcr})}%
{\label{prob:mfl}
  A set~$\mathcal{C}$ of~$n$ clients,
  a set~$\mathcal{F}$ of $m$ facilities,
  facility opening costs~$f\colon \mathcal{F}\to\Q_{\geq 0}$,
  and client service costs~$c\colon \mathcal{F}\times \mathcal{C}\to\Q_{\geq 0}$.}%
{Find a subset~$\mathcal{F'}\subseteq \mathcal{F}$ that minimizes
  \begin{align}
    \sum_{i \in \mathcal{F'}} f(i) + \sum_{j\in \mathcal{C}} \min_{i\in \mathcal{F'}} c(i,j).
    \label{eq:facloc}
  \end{align}}
When the cost function is a metric,
then the problem is also called \prob{\muflpTsc} (\prob{\muflpAcr}).
By showing that the goal function~\eqref{eq:facloc} is~\lin{},
we can prove:

\begin{lemma}%
  \label[lemma]{lem:magicweightsFL}
  There is an algorithm that, 
  on an input consisting of an instance~$(\mathcal{C},\mathcal{F},f,c)$ of \prob{\uflpAcr} and~$k\in\Q$,
  in time polynomial in~$|(\mathcal{C},\mathcal{F},f,c,k)|$ 
  computes 
  an instance~$(\mathcal{C},\mathcal{F},\bar f,\bar c)$ of \prob{\uflpAcr} and~$\bar k\in\Z$ such that
  \begin{enumerate}[(i)]
  \item $\norm{\bar{f}}{\infty} + \norm{\bar{c}}{\infty},|k|\leq 2^{4(nm+m+1)^3}(4(2n+m)+1)^{(nm+m+1)(nm+m+3)}$,
  \item any subset~$\mathcal{F'}\subseteq \mathcal{F}$ forms an optimal solution for $(\mathcal{C},\mathcal{F},f,c)$ if and only if $\mathcal{F'}$ forms an optimal solution for~$(\mathcal{C},\mathcal{F},\bar f,\bar c)$, and
  \item for any subset~$\mathcal{F'}\subseteq \mathcal{F}$ we have that
  $\sum_{i \in \mathcal{F'}} f(i) + \sum_{j\in \mathcal{C}} \min_{i\in \mathcal{F'}} c(i,j) \geq k \iff \sum_{i \in \mathcal{F'}} \bar f(i) + \sum_{j\in \mathcal{C}} \min_{i\in \mathcal{F'}} \bar c(i,j)\geq \bar k$.
  \end{enumerate}
\end{lemma}

\begin{proof}
	First, observe that~$f$ and~$c$ are \Zlin{1} as they can be represented as~$e_i^\top w$, where~$e_i$ denotes the unit vector with the~$i$th entry being one and~$w = (f(1),\ldots,f(m),\allowbreak c(1,1),\ldots, c(m,n))$ denotes a weight vector that contains all possible opening and serving costs.
	Since the goal function is composed of a sum of two sums, we will first analyze each of the sums individually and then analyze the outer sum.
	Observe that~$\sum_{i \in \mathcal{F'}} f(i)$ is~\Zlin{$m$} by \cref{lem:ft!Z!summaxmin}\eqref{alphZlin:sum} as~$|\mathcal{F'}| \leq m$.
	Similarly, since~$\min_{i\in \mathcal{F'}} c(i,j)$ is~\Zlin{$2$} by \cref{lem:ft!Z!summaxmin}\eqref{alphZlin:min}, it follows from \cref{lem:ft!Z!summaxmin}\eqref{alphZlin:sum} that $\sum_{j\in \mathcal{C}} \min_{i\in \mathcal{F'}} c(i,j)$ is~\Zlin{$2n$} as~$|\mathcal{C}|=n$.
	Next,
	we define
        \[f'(\ell,\mathcal{C},\mathcal{F'}) := 
		\begin{cases} \displaystyle
                  \sum_{i \in \mathcal{F'}} f(i) 						& \text{ if } \ell=1,\\
                  \displaystyle
			\sum_{j\in \mathcal{C}} \min_{i\in \mathcal{F'}} c(i,j) 	& \text{ if } \ell=2.
		\end{cases}\]
	Observe that~$f'$ is \Zlin{($2n+m$)} as it is \Zlin{($2n+m$)} in each of the two cases by \cref{obs:geqalphalin}.
	Moreover,
	note that the goal function can be represented as~$\sum_{\ell\in\{1,2\}} f'(\ell,\mathcal{C},\mathcal{F'})$.
	Due to \cref{lem:ft!Z!summaxmin}\eqref{alphZlin:sum},
	it follows that the goal function is~\Zlin{$2\cdot (2n+m)$}.
	Finally, 
	\cref{prop:ft!alphaZlin} yields the desired statement with~$\alpha = 2(2n+m)$ and~$d = nm+m$.
\end{proof}

We can apply \cref{lem:magicweightsFL} also for the \muflpAcr{},
which requires the cost function to satisfy the triangle inequality
$c(i,j)\le c(i,j') + c(i',j') + c(i',j) \text{ for all } i,i' \in \mathcal{C} \text{ and } j,j' \in \mathcal{F}$.
This easily follows from the following:
\begin{observation}
  \label[observation]{obs:ft!metric}
  Let~$w
  \in\Q^d$, $d\in\N$, and~$\K\in\{\Z,\Q\}$. 
  Then for every~$r\geq 4$ 
  and~$\wmod{w}
  \in\eqc{w}[r][\K^d]$ it holds that~\(\sign(w^\top (\vec{e}_i+\vec{e}_j+\vec{e}_k-\vec{e}_\ell))=\sign(\wmod{w}^\top (\vec{e}_i+\vec{e}_j+\vec{e}_k-\vec{e}_\ell))\) for each~\(i,j,k,l\in\{1,\dots,d\}\).
\end{observation}

The consequence of \cref{obs:ft!metric} is that if \cref{thm:FrankTardosRephr} is applied to a vector~$w$ encoding a metric 
(i.\,e.,
the entries of~$w$ are pairwise distances of some points), 
then the resulting vector~$\wmod{w}$ also encodes a metric. %
This property carries over
to \cref{prop:ft!alphaZlin,prop:ft!alphaQlin}.
Overall, we obtain 
the following result.
\begin{proposition}%
  \label{thm:uflp-small-weights}
  Each of \prob{\uflpAcr} and \prob{\muflpAcr} admits a problem kernel of size~$(n+m)^{O(1)}$.
\end{proposition}

This complements a result of Fellows and Fernau~\citep{FF11}
who showed a problem kernel with size exponential
in a given upper bound on the optimum
(which is unbounded in $n+m$).

\subsection{Scheduling with Tardy Jobs}
\label{sec:sched}
The parameterized complexity of scheduling problems
recently gained increased interest~\citep{MB18}.
In the following,
we demonstrate our framework on two single-machine
scheduling problems
where
the goal functions
are functions not of sets,
but of permutations,
so that the notions of linearity or additivity
do clearly not apply to them.

In the first problem,
we minimize the weighted number of tardy jobs.
Interestingly,
we are going to shrink
not only the weights of the jobs, 
but also their processing times and due dates,
where the goal function contains products of terms depending
on these numbers:

\optprob{Single-Machine Minimum Weighted Tardy Jobs (\prob{1$||\Sigma w_j U_j$})}%
{A set $J:=\{1,\dots,n\}$ of jobs,
  for each job $j\in J$
  a processing time~$p_j\in\N$,
  a due date~$d_j\in\N$,
  and a weight~$w_j\in\N$.}%
{Find a total order~$\preceq$ on~$J$ that minimizes
  the \emph{weighted number of tardy jobs}
  \begin{align*}
    \sum_{j\in J}w_jU_j,    &&\text{where}&&
                                             U_j:=
                                             \begin{cases}
                                               1&\text{if }d_j<C_j\\
                                               0&\text{otherwise,}
                                             \end{cases}
                                          &&\text{and}&&
                                                         C_j:=\sum_{i\preceq j}p_i.
  \end{align*}
}
In other words, $U_j$ is 1 if
job~$j$ is \emph{tardy}, that is, its \emph{completion time}~$C_j$ is after its due date~$d_j$.
The problem is weakly NP\hyp hard \citep{Kar72},
solvable in pseudo\hyp polynomial time \citep{LM69},
and is well\hyp studied in terms of parameterized complexity \citep{HermelinKPS21,HMO19,MB18},
yet there are no known kernelization results.

To apply our framework to \prob{1$||\Sigma w_j U_j$},
we show that we can also compose functions via case distinctions (like the one used to define~$U_j$) with linearizable constraints.
\begin{lemma}%
	\label[lemma]{lem:lin-constraints}
	Let~$f_1, f_2, g\colon L\times \Q^d\to \Q$.
	If~$f_1, f_2$, and~$g$ are \Klin{$\alpha$},
	then the following function is \Klin{$\alpha$}:
	\[h(x,w) = \begin{cases} f_1(x,w) &\text{if } g(x,w) \le 0, \\ f_2(x,w) & \text{otherwise}. \end{cases}\]
\end{lemma}

\begin{proof}
  Let~$w\in\Q^d$, 
  $x\in L$, 
  and~$\wmod{w}\in\eqc{w}[\alpha][\K^d]$. %
  We know that there exists a vector~$b_{x,w}\in\K^d_{\alpha}$ with~$\norm{b_{x,w}}{1}\leq \alpha$ such that
  $g(x,w')=b_{x,w}^\top w'$ for all~$w'\in\eqc{w}[\alpha][\K^d]$.
  It follows that
  \[ \sign(g(x,w)) = \sign(b_{x,w}^\top w) = \sign(b_{x,w}^\top \wmod{w}) = \sign(g(x,\wmod{w})). \]
  Thus, we have that
  $h(x,w) = f_1(x,w) \iff h(x,\wmod{w}) = f_1(x,\wmod{w})$ and~$h(x,w) = f_2(x,w) \iff h(x,\wmod{w}) = f_2(x,\wmod{w})$.
  Since both~$f_1$ and~$f_2$ are~\Klin{$\alpha$},
  the statement follows.
\end{proof}

Using \cref{lem:lin-constraints} and \cref{lem:ft!Z!summaxmin}\eqref{alphZlin:sum},
one can decompose the goal function~$\sum_{j\in J}w_jU_j$ for an order~$\preceq$ into simple linearizable functions.
Our framework then yields the following:

\begin{lemma}
	\label[lemma]{lem:wjuj}
	There is an algorithm that, 
	on an instance~$I$ of 1$||\Sigma w_jU_j$ and $k \in \N$,
	computes in polynomial time an instance~$I'$ of 1$||\Sigma w_jU_j$ and a~$k'$ such that
	\begin{enumerate}[(i)]
		\item each processing time, due date, weight, and~$k'$
		is
		at most
		$2^{4(3n+1)^3}(2 n^2+1)^{(3n+1)(3n+3)}$, 
		\item any solution~$\preceq$ is optimal for~$I$ if and only if it is optimal for~$I'$, and
		\item $I$ has a solution of cost at most~$k$ if and only if~$I'$ has a solution of cost at most~$k'$.
	\end{enumerate}
\end{lemma}

\begin{proof}
	We shrink the entries in the vector
	\(u=(w_1,\dots,w_n,p_1,\dots,p_n,d_1,\dots,d_n)^\top\in\N^{3n}.\)
	We can rewrite the goal function value for a solution~$\preceq$ as
	\[
		f(\preceq,u) = \sum_{j = 1}^{n} h(\preceq,j,u) \quad\text{ with }\quad h(\preceq,j,u) =
			\begin{cases}
				0 & \text{ if }d_j \ge C_j = \displaystyle\sum_{i\preceq j}p_i,\\
				w_j & \text{ otherwise ($d_j < C_j$)}.
			\end{cases}
	\]
	To this end, define~$g(\preceq,j,u) = - d_j + \sum_{i\preceq j}p_i$.
	Note that by \cref{obs:1linweightfcts} and \cref{lem:ft!Z!summaxmin}\eqref{alphZlin:sum} we have that~$g$ is \Zlin{$n$}.
	Hence, applying \cref{lem:lin-constraints} with the \Zlin{0}~$f_1 \equiv 0$ and the \Zlin{1}~$f_2(j,u) = w_j$ (see \cref{obs:1linweightfcts}) shows that~$h$ is~\Zlin{$n$}.
	Thus, by \cref{lem:ft!Z!summaxmin}\eqref{alphZlin:sum}, $f$ is \Zlin{$n^2$}.
	The statement of the lemma now follows from \cref{prop:ft!alphaZlin}.~\end{proof}

        We point out that a more careful and
        direct analysis shows that the goal function is even
        \Zlin{$1$}.
        However,
        we skip this here since it is more tedious.
Using \cref{lem:wjuj}(iii),
one gets the following.

\begin{proposition}%
	\label{prop:wjuj-kernel}
	1$||\Sigma w_jU_j$ admits a problem kernel of size polynomial in~$n$.
\end{proposition}

In the next problem,
one minimizes the total tardiness of jobs on a single machine.

\optprob{Single-Machine Minimum Total Tardiness (\prob{1$||\Sigma T_j$})}%
{A set $J:=\{1,\dots,n\}$ of jobs,
  for each job $j\in J$
  a processing time~$p_j\in\N$
  and
  a due date~$d_j\in\N$.}%
{Find a total order~$\preceq$ on~$J$ that minimizes
  the \emph{total tardiness}
  \begin{align*}
    \sum_{j\in J}T_j,    &&\text{where}&&
    T_j:=\max\{0,C_j-d_j\}
    &&\text{and}&&
    C_j:=\sum_{i\preceq j}p_i.
  \end{align*}  
}

Minimizing the total tardiness is motivated
by its equivalence
to minimizing the average tardiness
(just divide the goal function by~$n$).
The problem is fixed\hyp parameter tractable parameterized
by the maximum processing time \citep{Law77}
(this result was very recently
strengthened by \citet{KKL+19},
who showed fixed\hyp parameter
tractability even for the version
with parallel unrelated machines,
jobs with release dates and weights,
where jobs and machines
are given in a high\hyp multiplicity encoding
that encodes the numbers of jobs and machines of each type in binary).

It is easy to see that the goal function is a composition of sums and maxima. 
Hence, using \cref{lem:ft!Z!summaxmin}
one can show that the goal function is linearizable
and thus prove:
\begin{lemma}
	\label[lemma]{lem:tj}
	There is an algorithm that, on any input instance~$I$ of 1$||\Sigma T_j$ and~$k \in \N$, 
	in polynomial time 
	computes 
	an instance~$I'$ of 1$||\Sigma T_j$ and~$k'$ such that
	\begin{enumerate}[(i)]
		\item each processing time, due date, and~$k'$ is 
		at most
		$2^{4(2n+1)^3}(4 n^2+1)^{(2n+1)(2n+3)}$, 
		\item any solution~$\preceq$ is optimal for~$I$ if and only if it is optimal for~$I'$, and
		\item $I$ has a solution of cost at most~$k$ if and only if~$I'$ has a solution of cost at most~$k'$.
	\end{enumerate}
\end{lemma}

\begin{proof}
	We want to shrink weights in the vector $u=(p_1,\dots,p_n,d_1,\dots,d_n)^\top\in\N^{2n}$.
	To show that the goal function is linearizable, we express its value for a solution~$\preceq$ as
	\[
		f(\preceq,u) := \sum_{j\in J} f_T(\preceq,j,u) \quad\text{ with } \quad f_T(\preceq,j,u) := \max\{0,\sum_{i\preceq j}p_i-d_j\}.
	\]
  	By \cref{obs:1linweightfcts} and \cref{lem:ft!Z!summaxmin}\eqref{alphZlin:sum}, we have that~$g(\preceq,j,u) := \sum_{i\preceq j}p_i-d_j$ is \Zlin{$n$}.
  	Observe that~$h \equiv 0$ is \Zlin{0}.
  	Thus, by \cref{obs:geqalphalin} and \cref{lem:ft!Z!summaxmin}\eqref{alphZlin:max}, we have that~$f_T(\preceq,j,u)$ is \Zlin{$2n$}.
  	Using again \cref{lem:ft!Z!summaxmin}\eqref{alphZlin:sum}, we obtain that~$f(\preceq,u)$ is \Zlin{$2n^2$}.
  	The statement of the lemma now follows from \cref{prop:ft!alphaZlin}.
\end{proof}

From \cref{lem:tj}(iii) we get the following.
\begin{proposition}%
	\label{prop:sum-tj-kernel}
  1$||\Sigma T_j$ admits a problem kernel of size polynomial in~$n$.
\end{proposition}

\paragraph{Possible Generalizations}

The results in this section
can easily be generalized to
scheduling problems with parallel machines
(even with machine\hyp dependent processing times,
so\hyp called unrelated machines)
since the completion time~$C_j$ of job~$j$ can still be
represented as the sum of
processing times of predecessors of~$j$ on the same machine.
Variants with precedence constraints
(whose parameterized complexity is also well\hyp studied
\citep{BBB+16,BF95,FM03})
can be handled since
the completion time~$C_j$ of a job~$j$
can be expressed as the sum~$p_j$
and the completion time~$C_i$ of a direct
predecessor~$i$ of~$j$ in the precedence order
(possibly on a different machine).

Other goal functions can be handled as follows:
the total completion time $\sum_{j\in J}C_j$
is just the special case with $d_j=0$ for all jobs~$j\in J$.
Moreover,
one can replace the outer sums by maxima
to minimize the makespan or maximum tardiness.
We leave open whether 
\cref{prop:sum-tj-kernel}
can be proven for the weighted variant 1$||\Sigma w_jT_j$,
where each job~$j$ has a weight~$w_j$
and one minimizes $\sum_{j\in J}w_jT_j$:
in this case,
the goal function contains products of the weights we want to shrink.

\subsection{Arc Routing Problems with Min-Max Objective.}
\label{ssec:rpp}
Arc routing problems %
have applications in garbage
collection, mail delivery, meter reading, drilling,
and plotting
\citep{CL14}.
Their parameterized complexity is intensively studied
 \citep{BNSW15},
which 
led to promising results on real\hyp world instances
\citep{BFT20b,BKS17}.
Of particular interest are problem variants with multiple vehicles
with tours of balanced length \citep{AHL06,BCS14}, 
for example:

\optprob{\prob{Min-Max $k$-Rural Postman Problem} (\prob{MM $k$-RPP})%
}%
{An undirected graph~$G=(V,E)$, edge lengths~$c\colon E\to \N$,
  and a subset~$R\subseteq E$ of \emph{required edges}.}%
{Find closed walks $w_1,\dots,w_k$ in~$G$
  such that $R\subseteq\bigcup_{i=1}^kE(w_k)$
  that minimize
  \(
  \max\{c(w_i)\mid 1\leq i\leq k\},
  \)
  where $E(w_i)$ is the set of edges
  and $c(w_i)$~is the total length of edges on~$w_i$.
}

A key feature of the $k=1$ case (known as the
\prob{Rural Postman Problem}) is that one
can simply enforce the triangle inequality \citep{BHNS14}
and thus get an
equivalent instance with $2|R|$~vertices~\citep{BNSW15}. %
For \prob{MM $k$-RPP}, 
we \emph{partly} enforce the triangle inequality to
generalize this:

\begin{lemma}%
  \label[lemma]{lem:rpprule}
  In polynomial time,
  one can turn an instance~$(G,R,c)$ of \prob{MM $k$-RPP}
  into an instance~$(G',R,c^\triangledown)$
  on $3|R|$~vertices
  such that
  any solution for~$(G,R,c)$
  can be turned in polynomial time into a solution of at most
  the same cost for~$(G',R,c^\triangledown)$,
  and vice versa.
\end{lemma}

\begin{proof}
We first turn~$G$ into a complete graph~$G^*=(V,E')$
with edge lengths
\[
  c^\triangledown\colon E'\to\N,\{u,v\}\mapsto
  \begin{cases}
    c(\{u,v\})&\text{ if $\{u,v\}\in R$},\\
    \dist_c(u,v)&\text{ otherwise},
  \end{cases}
\]
where $\dist_c(u,v)$ is the length of a shortest $u$-$v$-path
in~$G$ according to~$c$.
Any feasible solution for $(G,R,c)$
is feasible for $(G^*,R,c^\triangledown)$ and has at most the same cost.
In the other direction,
one can replace non\hyp required edges in a feasible solution for~$(G^*,R,c^\triangledown)$
by shortest paths in~$G$ in polynomial time
to get a feasible solution of at most the same cost for~$(G,R,c)$.

Let $\bar V(R)\subseteq V$ be so that
for each edge $\{u,v\}\in R$,
it contains the vertices of at least one shortest $u$-$v$-path~$p$
(including $u$ and $v$).
Observe that if $p$~contains a vertex~$x$ not incident to any edge in~$R$,
then $(u,x,v)$ is a $u$-$v$-path of at most the same length.
Thus,
we can easily compute the set $\bar V(R)$ so that
$|\bar V(R)|\leq 3|R|$:
it contains the end points~$u$ and~$v$ for each edge~$\{u,v\}\in R$
and at most one vertex on a shortest $u$-$v$-path
not incident to edges in~$R$.

The key observation is now that
any shortest closed walk in~$G^*$
containing a subset~$R'\subseteq R$ of required edges
can be shortcut (in polynomial time)
so as to only contain vertices of~$\bar V(R)$.
Thus,
we can simply take~$G'=G^*[\bar V(R)]$.
\end{proof}

One can prove a problem kernel by shrinking the weights.

\begin{lemma}
 \label[lemma]{lem:rpp}
 There is an algorithm that,
 on an input instance~$I$ of \prob{MM $k$-RPP}
 with $m$~edges and~$\kappa\in\N$,
 computes in polynomial time
 an instance~$I'$ and~$\kappa'$ such that
 \begin{enumerate}[(i)]
  \item each edge cost is upper-bounded by $2^{4(m+1)^3}\cdot (8m+1)^{(m+1)(m+3)}$,
  \item a set of walks is an optimal solution for~$I$
    if and only if it is optimal for~$I'$, and
  \item $I$ has a solution of cost at most~$\kappa$ if and only if~$I'$ has a solution of cost at most~$\kappa'$.
 \end{enumerate}
\end{lemma}

\begin{proof}
  First observe that,
  without loss of generality,
  each walk~$w_i$ in a solution
  contains each edge of~$G$ at most two times:
  if it contains an edge~$e$ three times,
  then two occurrences of~$e$ can be removed,
  the walk gets shorter by~$2c(e)$,
  yet the edge~$e$ remains covered.
  Thus, one can write
  \(
    f(w_i,c):=c(w_i)=x^\top c,
  \)
  where the vector $c\in\N^m$ contains the edge costs
  and $x\in\{0,1,2\}^m$
  indicates how often each edge is on walk~$w_i$.
  Since for \emph{any} alternative edge weight vector~$c'$ we have $f(w_i,c')=x^\top c'$,
  it follows that $f$~is \Zlin{$2m$}.

  By \cref{lem:ft!Z!summaxmin},
  it follows that the goal function $\max_{1\leq i\leq k}f(w_i,c)$
  of \prob{MM $k$-RPP} is \Zlin{$4m$}
  and the lemma follows from \cref{prop:ft!alphaZlin}.
\end{proof}

\begin{proposition}
  \label{prop:rpp}
  \prob{Min-Max $k$-Rural Postman Problem} has a $3|R|$-vertex kernel
  of size polynomial in~$|R|$.
\end{proposition}

It is straightforward to transfer
\cref{lem:rpp} to other vehicle routing problems
that minimize maximum tour length.

\subsection{Power Vertex Cover}
\label{ssec:pvc}

Angel~\etal~\citep{ABEL18} claimed
a polynomial\hyp size problem kernel for the following %
problem
parameterized by the number of vertices that are assigned non-zero values
in a solution:
\optprob{\prob{Power Vertex Cover} (\prob{PVC})%
}%
{\label{prob:pvc}
An undirected graph~$G=(V,E)$ with edge weights~$w\colon E\to \Qnn$.}%
{Find an assignment~$\mu\colon V\to \Qnn$
minimizing
  $\sum_{v\in V} \mu(v)$
  such that, for each edge~$e=\{u,v\}\in E$, one has~$\max\{\mu(u),\mu(v)\}\geq w(e)$.
}

In fact,
Angel~\etal~\citep{ABEL18} only proved a partial kernel,
since the edge weights in the kernel can be arbitrarily large.
Using~\cref{prop:ft!alphaZlin},
we prove that we can shrink the weights.
To this end,
our application of the losing\hyp weight technique for \prob{PVC} relies on the following.

\begin{observation}\label[observation]{obs:powerVC}
 If $\mu$ is an optimal solution,
 then for every~$v\in V$, we have~$\mu(v)\in \{w(e)\mid e\in E\}\cup\{0\}$.
\end{observation}

This leads to the following equivalent problem formulation of~\prob{PVC}.
\optprob{\prob{Power Vertex Cover 2} (\prob{PVC2})%
}%
{\label{prob:pvc2}
An undirected graph~$G=(V,E)$ with edge weights~$w\colon E\cup\{\emptyset\}\to \Qnn$ with~$w(\emptyset)=0$.}%
{Find an assignment~$\mu\colon V\to E\cup\{\emptyset\}$ such that for each edge~$e=\{u,v\}\in E$ it holds true that~$\max\{w(\mu(u)),w(\mu(v))\}\geq w(e)$ and~$\mu$ minimizes
  $\sum_{v\in V} w(\mu(v))$.
}
\begin{lemma}
 \label[lemma]{lem:magicweightsPVC}
 There is an algorithm that,
 on an instance~$I=(G=(V,E),w)$ of \prob{PVC2} with $n:=|V|$ and $m:=|E|$, and~$k\in\Q$,
 in time polynomial in~$|(I,k)|$ 
 computes 
 an instance~$I'=(G=(V,E),\wmod{w})$
 of \prob{PVC2} and~$\wmod{k}\in\Z$ such that
 \begin{enumerate}[(i)]
  \item $\norm{\wmod{w}}{\infty},|\wmod{k}|\leq 2^{4(m+1)^3}\cdot (2n+1)^{(m+1)(m+3)}$,
  \item any assignment~$\mu\colon V\to E\cup\{\emptyset\}$
    forms an optimal solution for $I$ if and only if $\mu$ forms an optimal solution for~$I'$,
    and
  \item for any assignment~$\mu\colon V\to E\cup\{\emptyset\}$ it holds that $\sum_{v\in V} w(\mu(v))\leq k\iff\break \sum_{v\in V} \wmod{w}(\mu(v))\leq \wmod{k}$.
 \end{enumerate}
\end{lemma}

\begin{proof}
	Let~$f(v,w)=w(e)$ if~$\mu(v)=e$, and~$0$ if~$\mu(v)=\emptyset$.
	Due to~\cref{obs:1linweightfcts},
	$f$~is~\Zlin{$1$}.
	Hence, $g(V,w)=\sum_{v\in V} f(v,w)$ is~\Zlin{$n$}.
	\cref{prop:ft!alphaZlin} now yields the desired statement with~$\alpha = n$ and~$d = m$.
\end{proof}

Using \cref{lem:magicweightsPVC} and the partial kernel of Angel~\etal~\citep{ABEL18}, we obtain 
\cref{thm:pvc-small-weights}.

\begin{proposition}%
  \label{thm:pvc-small-weights}
  \prob{Power Vertex Cover} admits a polynomial kernel with respect to the number of non-zero values in a solution.
\end{proposition}

\subsection{Chamberlin-Courant Committee with Cardinal Utilities}
\label{ssec:cccc}
Another exemplary application is the following problem from computational social choice.
It deals with the Chamberlin-Courant voting rule~\citep{ChamberlinC1983},
which already has been studied from a parameterized complexity point of view~\citep{FluschnikSTW19,MisraSV17,SkowronF17}.

\dectaskprob{Chamberlin-Courant Committee with Cardinal Utilities (C$^4$U)%
}%
{A set~$V$ of voters, 
a set~$A$ of alternatives, 
  a function~$u\colon V\times A\to\Q_{\geq 0}$, and~$k\in\N$.}%
{Find a subset~$A'\subseteq A$ of size at most~$k$ that maximizes
  \begin{align}
    \sum_{v \in V} \max_{a\in A'} u(v,a).\label{gfcccu}
  \end{align}
}
We will show that the goal function~\eqref{gfcccu} is linearizable.

\begin{lemma}
 \label[lemma]{lem:magicweightscccu}
 There is an algorithm that,
 on an input consisting of an instance~$(V,A,u,k)$ of \prob{C$^4$U} with $n:=|V|$ and $m:=|A|$, and~$p\in\Q$,
 computes in time polynomial in~$|(V,A,u,k,p)|$ an instance~$(V,A,\bar u,k)$
 of C$^4$U and~$\bar{p}\in\Z$ such that
 \begin{enumerate}[(i)]
  \item $\norm{\bar{u}}{\infty},|\bar p|\leq 2^{4(nm+1)^3}\cdot (4n+1)^{(nm+1)(nm+3)}$,
  \item any subset~$A'\subseteq A$
    forms an optimal solution for $(V,A,u,k)$ if and only if $A'$ forms an optimal solution for~$(V,A,\bar u,k)$,
    and
  \item for any subset~$A'\subseteq A$ we have that $\sum_{v \in V} \max\limits_{a\in A'} u(v,a)\geq p \iff \sum_{v \in V} \max\limits_{a\in A'} \bar u(v,a)\geq \bar p$.
 \end{enumerate}
\end{lemma}

\begin{proof}
  Observe that the goal function can be restated as follows:
  \[ \sum_{v \in V} \max_{a\in A'} u(v,a) = \sum_{v \in V} \max_{(v,a)\in \{v\}\times A'} u(v,a). \]
  Due to~\cref{obs:1linweightfcts},
  $u\colon V\times A\to \Q_{\geq 0}$ is \Zlin{1}.
  Note that the weight vector representing~$u$ is of dimension~$d=nm$.
  By~\cref{lem:ft!Z!summaxmin}\eqref{alphZlin:max},
  we know that $g(v,A')=\max_{(v,a)\in \{v\}\times A'} u(v,a)$ is \Zlin{2}.
  Finally,
  by \cref{lem:ft!Z!summaxmin}\eqref{alphZlin:sum},
  $h(V,A')=\sum_{v\in V} g(v,A')$ is~\Zlin{$2n$}.
  By~\cref{prop:ft!alphaZlin} with~$\alpha=2n$ and~$d=nm$,
  the claim follows.
\end{proof}

\Cref{lem:magicweightscccu} yields 
\cref{thm:c4u-small-weights}.

\begin{proposition}%
  \label{thm:c4u-small-weights}
  \prob{Chamberlin-Courant Committee with Cardinal Utilities} admits a problem kernel of size polynomial in the combined parameter 
  number of voters and alternatives.
\end{proposition}

\section{Concluding Remarks}
\label{sec:concl}

The losing\hyp weight technique due to Frank and Tardos~\citep{ft87} is a key tool to obtain polynomial problem kernels for weighted parameterized problems.
While Marx and V\'egh~\citep{MV15} and Etscheid~\etal~\citep{ekmr17} proved the usefulness of the technique for several problems with additive goal functions,
we demonstrated its applicability for %
a larger class of functions (linearizable functions) 
containing next to additive also non-additive functions.
In addition,
in~\cref{ssec:apps} we displayed our recipe DDD 
to be a neat manual 
for applying the losing-weight technique to the class of 
(linearizable) 
functions.

As Etscheid~\etal~\citep{ekmr17} pointed out, 
one direction for future work is to improve the upper bound in~\cref{thm:FrankTardosRephr} on the maximum norm of the output vector.
In this direction,
Eisenbrand~\etal~\citep{EHK+19} recently proved a stronger upper bound,
yet non-constructively.
Another direction, seemingly not addressed so far, aims
for a better running time:
Frank and Tardos~\citep{ft87} state no explicit running time of their algorithm, and Lenstra~\etal~\citep[Proposition~1.26]{LLL82} state that their simultaneous Diophantine approximation algorithm, which forms a subroutine in Frank and Tardos' technique, runs in $d^6\cdot\tlog{\norm{w}{\infty}}^{O(1)}$~time.
This is clearly a bottleneck for the practical applicability
of the techniques we discussed.
Hence,  we put forward the following:
Can~\cref{thm:FrankTardosRephr} be executed in quadratic, or even linear time?
We point out that for approximate kernelizations,
there is an analog to~\cref{thm:FrankTardosRephr} executable in linear time~\citep{BFT20b}.

\subparagraph*{Acknowledgments.}
We thank Oxana Yu.~Tsidulko for pointing 
out how to apply the Frank-Tardos theorem to the \prob{Uncapacitated Facility Location} problem (\cref{sec:uflp}).
We thank the anonymous reviewers from 
\emph{Discrete Applied Mathematics}
for their constructive feedback.
In memory of Rolf Niedermeier, 
our colleague, 
friend, 
and mentor, 
who sadly passed away before this paper was published.

\subparagraph*{Funding.}
Till Fluschnik was supported by the
DFG project TORE (NI~369/18).

\bibliographystyle{kernel_symcon}
\bibliography{b-short}

\end{document}